\documentclass[journal,12pt,draftclsnofoot,onecolumn]{IEEEtran}
\usepackage{color}
\usepackage{balance}
\usepackage{graphicx}
\usepackage{enumerate}
\usepackage{subfigure}
\usepackage{algorithm}
\usepackage{algorithmic}
\usepackage{booktabs}
\usepackage{hyperref}
\hypersetup{colorlinks=true,
	linkcolor=black,
	citecolor=black}
\usepackage{cite}
\usepackage{amssymb}
\usepackage{amsmath}
\usepackage{authblk}
\usepackage{stfloats}
\usepackage[utf8]{inputenc}
\usepackage{amsthm}
\usepackage{cases}
\newtheoremstyle{note}
{3pt}
{3pt}
{\itshape}
{1em}
{\itshape}
{:}
{.5em}
{}

\newtheorem{theorem}{Theorem}

\newtheorem{corollary}{Corollary}
\renewenvironment{proof}{{$\quad $\it Proof:}$\ $}{$\hfill\blacksquare$}

\begin{document}
	\vspace{-10pt}
\title{Advanced NOMA Assisted Semi-Grant-Free Transmission Schemes for Randomly\\ Distributed Users \vspace{-8pt}}
\vspace{-10pt}
	
%
\vspace{-10pt}
\author{
Huabing Lu,~\IEEEmembership{Member,~IEEE,}
Xianzhong Xie,~\IEEEmembership{Member,~IEEE,}\\
\vspace{-17pt}
Zhaoyuan Shi,~\IEEEmembership{Member,~IEEE,}
Hongjiang Lei,~\IEEEmembership{Senior Member,~IEEE,}\\ \vspace{-5pt}
Helin Yang,~\IEEEmembership{Member,~IEEE,}
and Jun Cai,~\IEEEmembership{Senior Member,~IEEE}
		        \vspace{-35pt} 
\thanks{H. Lu is with the Key Laboratory of Intelligent Control and Optimization for Industrial Equipment of Ministry of Education, Dalian University of Technology, Dalian 116024, China, and also with the School of Computer Science and Technology, Chongqing University of Posts and Telecommunications, Chongqing 400065, China (e-mail: luhuabing@dlut.edu.cn).}
\thanks{X. Xie is with the School of Computer Science and Technology, Chongqing University of Posts and Telecommunications, Chongqing 400065, China (e-mail: xiexzh@cqupt.edu.cn).}
\thanks{Z. Shi is with the Key Laboratory of Intelligent Perception and Computing of Anhui Province, Anqing Normal University, Anqing 246011, China (e-mail: shizy@stu.cqupt.edu.cn).}

\thanks{H. Lei is with the School of Communication and Information Engineering, Chongqing University of Posts and Telecommunications, Chongqing 400065, China (e-mail: leihj@cqupt.edu.cn).}
\thanks{H. Yang is with the Department of Information and Communication Engineering, School of Informatics, Xiamen University, Xiamen 361005, China (e-mail: helinyang066@xmu.edu.cn).}
\thanks{J. Cai is with the Network Intelligence and Innovation Lab (NI$^2$L), Department of Electrical and Computer Engineering, Concordia University, Montreal, QC H3G 1M8, Canada (e-mail: jun.cai@concordia.ca).}
\vspace{-5pt}

	}
	
	
	\maketitle
	\vspace{-8pt}	
\begin{abstract}
	\vspace{-5pt}
Non-orthogonal multiple access (NOMA) assisted semi-grant-free (SGF) transmission has recently received significant research attention due to its outstanding ability of serving grant-free (GF) users with grant-based (GB) users' spectrum, which {\color{black}greatly improves} the spectrum efficiency and effectively {\color{black}relieves} the massive access problem of 5G and beyond networks. In this paper, we first study the outage performance of the {\color{black}greedy} best user scheduling SGF scheme (BU-SGF) by considering the impacts of Rayleigh fading, path loss, and random user locations. {\color{black}In order to tackle the admission fairness problem of the BU-SGF scheme}, we propose a fair SGF scheme by applying cumulative distribution function (CDF)-based scheduling (CS-SGF), {\color{black}in which the GF user with the best channel relative to its own statistics will be admitted}. Moreover, by employing the theories of order statistics and stochastic geometry, the outage performances of both BU-SGF and CS-SGF schemes are analyzed. Theoretical results show that both schemes can achieve full diversity orders {\color{black}only} when the served users' data rate is capped, which severely {\color{black}limits} the rate performance of SGF schemes. To further address this issue, we propose a distributed power control strategy to relax such data rate constraint, and derive {\color{black}analytical expressions} of the two schemes' outage performances under this strategy. Finally, simulation results validate the fairness performance of the proposed CS-SGF scheme, the effectiveness of the power control strategy, and the accuracy of the theoretical analyses.
		
\textbf{\emph{Index Terms}} ---NOMA, semi-grant-free, CDF-based scheduling, outage probability, fairness.
	\end{abstract}
	
	\IEEEpeerreviewmaketitle
\section{Introduction}
	\renewcommand{\baselinestretch}{0.75} 

 {\color{black}\IEEEPARstart{W}{ith} the fast development of {\color{black}the} Internet of Things (IoT), more and more devices are expected to be connected to the networks. It is predicted that the number of {\color{black}connected devices} will reach 31.4 billion by 2023, and more than 60\% of which will be IoT connections \cite{2018_Ericsson}. These large {\color{black}number} of IoT devices arouse a paradigm shift from the current human-type communication oriented systems,} where the {\color{black}packets} are always quite long and transmitted based on grant-based (GB) protocols (namely, each communication device first transmits a scheduling request to the base station (BS) and then the BS sends a resource allocation grant back). On the contrary, short packets are common for the traffic generated by IoT devices \cite{Shirvanimoghaddam_2017_massive NOMA,2016_Durisi_short packets,2017_Dawy_ Toward mMTC}, which is unsuitable to be transmitted with conventional GB protocols, since the lengthy request-grant process will be prohibitively costly for the signaling overhead and unacceptable as well for the {\color{black}resulting} latency in delay-critical IoT applications \cite{2018_Gharbieh_spatiotemporal,cui_2020_GFNOMA}. This motivates the development of Grant-free (GF) transmissions, where the request-grant process is omitted and some dedicated resource blocks are reserved for these IoT devices to transmit whenever the packets arrive \cite{cui_2020_GFNOMA,2014_Bayesteh_blind detectioni of SCMA,2019_Vaezi_multiple access,Lien_2017_comm}. By applying GF schemes, the transmission delay resulting from {\color{black}the} request-grant process is eliminated, and also the spectrum efficiency can be effectively improved. Nevertheless, without central access control of the BS, collisions may frequently occur in GF transmissions, since the spectrum reserved for GF transmission is ordinarily limited and it is inevitable that multiple users will choose the same resource in massive amount of IoT scenarios. 

{\color{black}Integrating Non-orthogonal multiple access (NOMA) with GF transmission is a promising solution to this problem, by which multiple devices could transmit their signals using the same resource with different power levels or codebooks \cite{Yuan_2016_comm,2019_Dogan_NOMA GF,2019_Yang_sustainable NOMA}. In this paper, we focus on power-domain NOMA \cite{2017_Islam_PNOMA_survey,2017_JSAC_Ding_survey NOMA}. For comprehensive review of code-domain GF NOMA, we refer the readers to \cite{2018_Dai_Survey NOMA,2020_Elbayoumi_NOMA UDN,2020_Shahab_IoT survey} and the references therein. The combination of power domain NOMA with random access (NOMA-RA) scheme for multichannel ALOHA was first discussed in \cite{2017_Choi_ALOHA}, where users can choose the predetermined power levels for uplink transmission. Further, a layered random access scheme was proposed to enhance the throughput of multichannel ALOHA in \cite{2018_Choi_Layered}. By making use of channel inversion technique, the received power levels were set as two \cite{2018_Seo_TVT} and multiple \cite{2018_Seo_coml} target values for NOMA-RA systems, and it is shown in \cite{2021_Yu_TWC} that increasing the number of power levels may further improve the throughput gain. Whereas, only successive interference cancellation (SIC) decoding strategy was considered in \cite{2017_Choi_ALOHA,2018_Choi_Layered,2018_Seo_TVT,2018_Seo_coml,2021_Yu_TWC}. Based on slotted ALOHA (SA) and NOMA (SA-NOMA), \cite{2020_Tegos_TCOM} investigated the performance of both SIC and joint decoding (JD) strategies for wireless sensor networks, which showed that JD could effectively avoid outage error floors and SA-NOMA outperforms SA.} However, in systems with a {\color{black}large quantity} of devices, the number of connections may still exceed the NOMA capability for successful decoding, which deteriorates the system performance \cite{2019_Yang_sustainable NOMA}.

{\color{black}To alleviate this situation}, NOMA assisted semi-grant-free (SGF) transmission schemes, which encourage GB users to share their resources with delay-tolerant IoT devices  {\color{black}transmitting with GF protocols, have received much research attention \cite{2019_Ding_Semi-GF,2019_Yang_sustainable NOMA}. In SGF transmission schemes, some IoT devices with delay-tolerant packets can be bypassed with the resources which would otherwise be solely utilized by the GB users, so that both the connectivity and spectral efficiency can be improved. {\color{black}As} the number of devices that content for the dedicated resources reserved for GF transmissions is reduced, the successful communication probability for the latency-critical IoT devices could be improved.} Compared to pure GB transmission, lower signaling overhead is needed in SGF transmission, meanwhile the collision event is effectively managed compared with pure GF transmission \cite{2019_Ding_Semi-GF}.


\subsection{Related Work}
The NOMA-assisted SGF transmission schemes have been investigated in \cite{2019_Ding_Semi-GF,2020_Yang_APA,2020_Jayanth_PIMRC,2020_Chao_Semi-GF,Zhang_2020_WCL,2020_Ding_New_SGF}. Specifically, the concept of SGF {\color{black}transmission} was firstly proposed in \cite{2019_Ding_Semi-GF}, where two SGF mechanisms were developed to restrict the number of admitted GF users and ensure that the admission of the GF users {\color{black}does} not cause too much performance degradation to the GB users. {\color{black}Based on the framework of \cite{2019_Ding_Semi-GF}, in order to improve the admitted GF users' outage performances, Yang \textit{et al.} \cite{2020_Yang_APA} proposed an adaptive power allocation strategy to restrict the transmit power of the GB user, as the GB user's signal was assumed to be always decoded at the second stage of SIC. Only Rayleigh fading was considered in \cite{2019_Ding_Semi-GF,2020_Yang_APA}, while the impacts of path loss and user locations were not taken into account. In this regard, Jayanth \textit{et al.} \cite{2020_Jayanth_PIMRC} considered a homogeneous user distribution scenario, and two GF users with the largest and second largest channel gains were admitted. To ensure the performance of the GB user, they exploited the principle of underlay cognitive radio to restrict the interference generated by the two admitted GF users within a threshold. Note that, the distances of all GF users to the BS are assumed to be {\color{black}the} same in \cite{2020_Jayanth_PIMRC}. To step further, Zhang \textit{et al.} investigated the spatial effect of random user locations on the performance of SGF schemes \cite{2020_Chao_Semi-GF,Zhang_2020_WCL}. They proposed a dynamic threshold protocol for the admission of the GF users and analyzed the outage performance \cite{2020_Chao_Semi-GF} and  ergodic rate \cite{Zhang_2020_WCL} for the randomly scheduled GF users by employing stochastic geometry.
	
In the aforementioned studies \cite{2019_Ding_Semi-GF,2020_Yang_APA,2020_Jayanth_PIMRC,2020_Chao_Semi-GF,Zhang_2020_WCL}, pre-fixed SIC decoding orders were assumed, which led {\color{black}to} performance degradation for the GB user or outage error floors for the GB/GF users. To this end, Ding \textit{et al.} \cite{2020_Ding_New_SGF} proposed a new SGF scheme by resorting to hybrid successive interference cancellation (HSIC) \cite{2020_Ding_HSIC}, in which the outage error floors could be effectively avoided if the product of the GB and GF users' target signal-to-interference-plus-noise ratio (SINR) is less than 1, while the performance of the GB user could still be guaranteed as it solely occupies the channel. In other words, for the new SGF scheme, the admission of GF users can effectively improve the spectrum efficiency without affecting the GB user's performance.}

\subsection{Motivation and Contributions}
Even though the aforementioned work has presented some innovative studies on NOMA-assisted SGF transmission schemes, three critical problems are still waiting for solution: \textit{1)} Only small scale fading (Rayleigh fading) was considered in \cite{2019_Ding_Semi-GF,2020_Yang_APA,2020_Jayanth_PIMRC,2020_Ding_New_SGF}, and the impact of user locations was {\color{black} not taken into consideration}. Although the impact of user locations was studied in \cite{2020_Chao_Semi-GF} and \cite{Zhang_2020_WCL}, they failed to exploit multi-user diversity since random user scheduling was considered; \textit{2)} Most of the existing SGF schemes preferred to schedule the GF users with the strongest or weakest channel gains,  {\color{black}in order to ensure the performance of the GB user \cite{2019_Ding_Semi-GF} or maximize the rate performance of the scheduled GF users \cite{2019_Ding_Semi-GF,2020_Yang_APA,2020_Jayanth_PIMRC,2020_Ding_New_SGF}}. This scheduling strategy may lead to the fairness\footnote{In general, fairness in user scheduling can be divided into two categories \cite{2015-Jin-Fundamental_limits_CDF}: throughput-based fairness and resource-based fairness. This paper focuses on resource-based fairness, namely, the GF users with different distances to the BS have an equal probability in scheduling.} problem, since the users closer to the BS (cell boundary) may be scheduled more frequently due to smaller (larger) path loss. {\color{black}In practice, both sum rate and access fairness are critical system performance indicators in wireless networks \cite{2013_Sediq_TWC,2014_Shi_CST}, especially for opportunistic scheduling \cite{2005_Tse}. Hence, it is necessary to develop {\color{black}an} SGF scheme {\color{black}} can ensure fair admission chance for the GF users with {\color{black}a} reasonable rate performance guarantee;} \textit{3)} {\color{black}In the new SGF scheme \cite{2020_Ding_New_SGF}}, robust transmission of the GF users (namely, the GF users can achieve non-zero diversity orders) can be achieved only in the case of {\color{black}{capped}} target rate pairs. It is necessary to develop {\color{black}an} SGF scheme {\color{black}} can guarantee robust transmission in all cases. 

Motivated by the previous discussions, this paper dedicates to {\color{black}investigating} the SGF transmission schemes in a more practical scenario of considering the impact of spatial user locations. To be specific, we first {\color{black}extend} the analysis on the outage performance of {\color{black}the} SGF scheme with best user scheduling (BU-SGF) (namely, the GF user with the maximal data rate will be scheduled \cite{2020_Ding_New_SGF}) under a channel model consisting of Rayleigh fading, path loss, and random user locations by employing stochastic geometry. After that, to address the admission fairness issue inherent in {\color{black}the} BU-SGF scheme, we introduce cumulative distribution function (CDF)-based scheduling to SGF scheme (termed as CS-SGF scheme), where the GF user whose channel condition is at its best state and most unlikely to be better will be admitted. Compared to the existing fair schemes, i.e., random selection SGF schemes \cite{2020_Chao_Semi-GF,Zhang_2020_WCL}, the proposed CS-SGF scheme can effectively exploit multi-user diversity. Moreover, we develop a power control strategy, which can relieve the restrictions on the users' target rates for achieving full diversity orders and further enhance the outage performance. The main contributions of this paper are outlined as follows:

$\bullet$ We develop a tractable performance analysis framework for {\color{black}the} BU-SGF scheme integrating Rayleigh fading, path loss, and spatial user locations. The proposed framework can be easily extended to analyze the performance of other fading channel models, such as Nakagami-$m$ fading, Rician fading, and so on.

$\bullet$ We propose a fair admission scheme for SGF transmission systems by invoking CDF-based scheduling, which can effectively utilize multi-user diversity.

$\bullet$ For both BU-SGF and CS-SGF schemes, we analyze the outage performances of the admitted GF users by applying order statistics and stochastic geometry. {\color{black}Meanwhile, we also derive the rate constraints under which both schemes can achieve full or zero diversity orders and identify the causes resulting in error floors.}

$\bullet$ {\color{black}We propose a distributed power control strategy for both GB and GF users to eliminate outage error floors and enhance the admitted GF users' outage performance. We then re-evaluate the outage performances of BU-SGF and CS-SGF schemes with the proposed power control strategy and show insights on why the error floors can be avoided after applying the proposed power control strategy.}



{\color{black}Our work extends the BU-SGF scheme \cite{2020_Ding_New_SGF} in the following several aspects: \textit{1)} Only Rayleigh fading is considered in \cite{2020_Ding_New_SGF}, while our work takes both the path loss and the two-dimensional spatial locations of the users into consideration. Meanwhile, a more general outage performance analysis framework is proposed which can be adapted to other channel models with the impact of both path loss and spatial user locations; \textit{2)} The condition on the theoretical analysis in \cite{2020_Ding_New_SGF}, i.e., the product of the two users' target SINRs is less than 1, is removed to make our analysis more general; \textit{3)} The outage performance of the traditional BU-SGF scheme is re-evaluated by considering the power control strategy.}
 
\subsection{Organization}
The rest of this paper is organized as follows. Section \ref{system model} introduces the system model. Two SGF transmission schemes are presented in Section \ref{two schemes}. In Section \ref{Performance of two schemes}, the outage performances of the two schemes with fixed transmit power are analyzed. Section \ref{Propose PC and performance analyze} proposes the power control strategy and analyzes the outage performances of the two schemes after applying the power control strategy. In Section \ref{Simulation and discussion}, simulation results are presented to verify the theoretical analyses and Section \ref{conclusion} concludes this paper.



	
\section{System Model}\label{system model}
In this section, the signal model is presented first, then the decoding scheme with HSIC and performance metric are introduced.

\begin{figure}[t]
	\centering
	\includegraphics[width=0.6\textwidth]{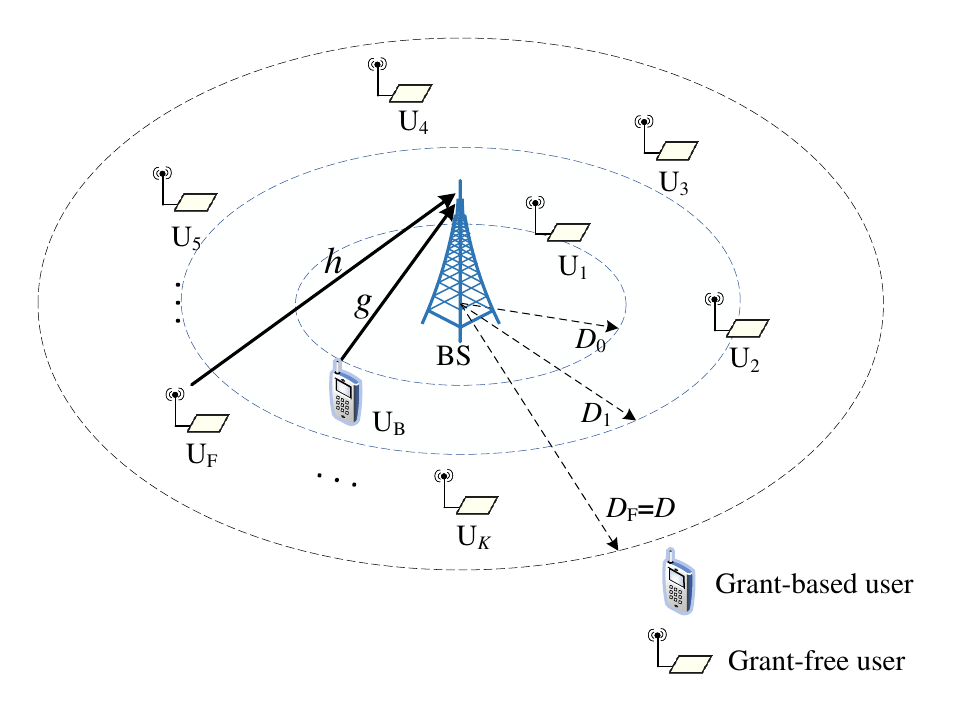}
			\vspace{-10pt}
	\caption{\color{black}An illustration of the considered SGF system model.}
	\label{fig1_system model}
	\vspace{-15pt}
\end{figure}

\subsection{Signal Model}
{\color{black}As shown in Fig. 1,} a single-cell uplink cellular network is considered, where the BS is located at the center of the coverage disc area with radius $D$. Similar to \cite{2019_Ding_Semi-GF,2020_Yang_APA,2020_Ding_New_SGF}, we consider $K$ GF users {\color{black}(denoted as U$_k$, $k\in\{1,\dots,K\}$)} and one GB user {\color{black}(denoted as $\text{U}_\text{B}$)}, where these GF users are distributed in the disc area $\mathcal{D}_\text{F}$ with radius $D_\text{F}$ $(D_\text{F}\leq D)$, and the GB user is randomly deployed in a ring region $\mathcal{D}_\text{B}$ with inner radius $D_0$ ($D_0\geq0$) and outer radius $D_1$ ($D_1\leq D$). The GF users are randomly distributed according to homogeneous Binomial point process (HBPP) \cite{2012_Haenggin_stochastic geometry}, namely, the GF users are randomly deployed within the coverage area of the BS. The GB user communicates with the BS in conventional grant-based protocol and has been allocated one specific resource block. Thus, we assume that the channel state information (CSI) of the GB user is known to the BS \cite{Zhang_2020_WCL}. {\color{black}Similar to \cite{2020_Ding_New_SGF}, in each time slot, one GF user, denoted as U$_{\text{F}}$ $(\text{U}_{\text{F}}\in \{\text{U}_1,\dots,\text{U}_K\})$, will be admitted for transmission using the resource block allocated to U$_{\text{B}}$ after distributed contention\footnote{Distributed contention has been widely used in opportunistic carrier sensing\cite{2020_Ding_New_SGF,2005_Zhao_opportunistic carrier sensing,2006_Bletsas_network_path_selection}, where the BS can schedule the most preferred user in a distributed manner. Take the strategy proposed in \cite{2006_Bletsas_network_path_selection} as an example, which selects the user with the best channel condition to transmit. After estimating the CSI, each user selects a backoff time, e.g., $\pi_k$ for the $k$-th user, which \textcolor{black}{decreases monotonously with increasing} channel gain. Once the contention time window (with duration $\pi_0$) begins, the $k$-th user will send a flag to the BS after $\pi_k$ $(\pi_k<\pi_0)$ expires. Thus, the user with the best channel will send its flag first and hence be identified to the BS.}, where the contention criteria will be specified in Section \ref{two schemes}.}


{\color{black}It is assumed that all nodes are equipped with a single antenna.} We consider a composite channel model with both quasi-static Rayleigh fading and large scale path loss, and the channel coefficients are assumed to be invariant during each time slot and change independently between slots. {\color{black}At the beginning of each time slot, we assume that each user can estimate its CSI by exploiting pilot signals sent by the BS.} The channel between the $k$-th GF user U$_k$ and the BS is modeled as $h_k=\frac{\zeta_k}{\sqrt{1+r_k^{\alpha}}}$, where $r_k$ represents the distance between U$_k$ and the BS, $\alpha$ denotes the path loss exponent, and $\zeta_k$ represents Rayleigh fading coefficient with {\color{black}$\zeta_k\sim\mathcal{CN}(0,1)$. Without} loss of generality, we assume that the GF users' channel gains are ordered as\footnote{Note that, this assumption is used to facilitate performance analysis, and all nodes in the system do not know this order \cite{2020_Ding_New_SGF}.}
\begin{equation}\label{channel order}
	\begin{aligned}
	|h_1|^2\leq\dots\leq|h_K|^2.
	\end{aligned}
\end{equation}
 Similarly, the channel of U$_{\text{B}}$ to the BS is defined as $g=\frac{\zeta_\text{B}}{\sqrt{1+r_{\text{B}}^{\alpha}}}$, where $\zeta_\text{B}\sim\mathcal{CN}(0,1)$ and $r_{\text{B}}$ denotes the distance between U$_{\text{B}}$ and the BS.

Based on these assumptions, the CDFs of the unordered channel gains of the GF and GB users can be respectively expressed as \cite{2014_Ding_random_deployed,2020_Lu_TVT}
\begin{equation}\label{CDF_F_F}
\begin{aligned}
F_\text{F}(x)
=&\frac{2}{D_\text{F}^2}\int_{0}^{D_\text{F}}\left[1-e^{-(1+r^\alpha)x}\right]rdr
\overset{(a)}{\approx}\frac{1}{2}\sum_{l=1}^{L}\Psi_l\left(1-e^{-\mu_lx}\right),\\
\end{aligned}
\end{equation}
\begin{equation}\label{CDF_unordered}
\begin{aligned}
F_\text{B}(y)
=&\frac{2}{D_1^2-D_0^2}\int_{D_0}^{D_1}\left[1-e^{-(1+r^\alpha)y}\right]rdr
\overset{(b)}{\approx}\frac{1}{D_1+D_0}\sum_{n=1}^{N}\Phi_n\left(1-e^{-c_ny}\right),\\
\end{aligned}
\end{equation}
where $\psi_l=\text{cos}\left(\frac{2l-1}{2L}\pi \right)$, $\mu_l=1+\left(\frac{D_\text{F}}{2}+\frac{D_\text{F}}{2}\psi_l\right)^\alpha$, $\Psi_l=\frac{\pi}{L}\sqrt{1-\psi_l^2}(1+\psi_l)$, $\varphi_n=\text{cos}\left(\frac{2n-1}{2N}\pi \right)$, $\phi_n=\frac{D_1+D_0}{2}+\frac{D_1-D_0}{2}\varphi_n$, $\Phi_n=\frac{\pi}{N}\sqrt{1-\varphi_n^2}\phi_n$, $c_n=1+\phi_n^\alpha$, $L$ and $N$ are parameters for ensuring complexity-accuracy trade-off. {\color{black}Due to the integrals in (\ref{CDF_F_F}) and (\ref{CDF_unordered}) can not be calculated in many communication scenarios (i.e., $\alpha>2$) \cite{2014_Ding_random_deployed}, the Gaussian-Chebyshev quadrature \cite{Gaussian_Chebyshev} is applied to perform approximation operations in $(a)$ and $(b)$.} From (\ref{CDF_unordered}), we can derive the probability density function (pdf) of U$_{\text{B}}$'s channel gain as

\begin{equation}\label{pdf_GB}
\begin{aligned}
f_\text{B}(y)
\approx\frac{1}{D_1+D_0}\sum_{n=1}^{N}\Phi_nc_ne^{-c_ny}.\\
\end{aligned}
\end{equation}

{\color{black} In each time slot,} {\color{black}the BS will receive the superimposed signals transmitted from the GB user U$_\text{B}$ and the admitted GF user U$_\text{F}$, which can be expressed as}
\begin{equation}
\begin{aligned}	
y=h\sqrt{P_{\text{F}}}s_{\text{F}}+g\sqrt{P_{\text{B}}}s_{\text{B}}+n_0,
\end{aligned}
\end{equation}
{\color{black}where $s_\nu$ ($\nu\in\{\text{B},\text{F}\}$) denotes the transmit symbol of $\text{U}_\nu$ with $\mathbb{E}\{|s_\nu|^2\}=1$. $P_\nu$ represents the transmit power of $\text{U}_\nu$, $h$ ($g$) denotes the channel from $\text{U}_{\text{F}}$ ($\text{U}_{\text{B}}$) to the BS,} and $n_0$ represents the additive white Gaussian noise (AWGN) with zero mean and variance {\color{black}$\sigma^2$}. To facilitate the theoretical analysis, we assume that all the GF users have the same target data rate. Let $R_{\text{B}}$ $(R_{\text{F}})$ and $\gamma_{\text{B}}=2^{R_{\text{B}}}-1$ $(\gamma_{\text{F}}=2^{R_{\text{F}}}-1)$ respectively represent the target data rate and the target SINR of U$_{\text{B}}$ ($\text{U}_{\text{F}}$). Assuming the maximal transmit powers of the GB and GF users are the same and denoted as $P_m$. The notations used in this paper are summarized in Table \ref{table1}.

\begin{table}[t!]
	\centering
	\caption{List of notations}
	\renewcommand\arraystretch{1.4} 
	\begin{tabular}{|l|l|}		
		\hline
		\textbf{Notation} & \textbf{Description}\\
		\hline
		$K$ & Number of the GF users\\
		$\mathcal{D}_\text{F}\ (\mathcal{D}_\text{B})$& Disc (ring) region distributed with GF (GB) users\\
		$D$&Coverage area of the BS\\
		$D_\text{F}$&Radius of the disc region $\mathcal{D}_\text{F}$\\
		$D_0\ (D_1)$&Inner (outer) radius of the ring region $\mathcal{D}_\text{B}$\\
		$\text{U}_{k}$&The $k$-th GF user $(k\in\{1,\dots,K\})$\\
		$\text{U}_{\text{F}}$&The admitted GF user $(\text{U}_{\text{F}}\in \{\text{U}_1,\dots,\text{U}_K\})$\\
		U$_{\text{B}}$&The GB user\\		
		$\alpha$&The path loss exponent\\
		$h_k$&Channel of $\text{U}_k$ with the BS\\
		$h\ (g)$&Channel of $\text{U}_{\text{F}}$ (U$_{\text{B}}$) with the BS\\
		$r_k$&Distance between $\text{U}_k$ and the BS\\
		$r_\text{F}\ (r_\text{B})$&Distance between $\text{U}_{\text{F}}$ (U$_{\text{B}}$) and the BS\\
		$R_{\text{F}}\ (R_{\text{B}})$ & Target rate of $\text{U}_{\text{F}}$ (U$_{\text{B}}$) \\
		$\gamma_{\text{F}}\ (\gamma_{\text{B}})$ & Target SINR of $\text{U}_{\text{F}}$ (U$_{\text{B}}$), $\gamma_{\text{F}}=2^{R_{\text{F}}}-1$ $(\gamma_{\text{B}}=2^{R_{\text{B}}}-1)$ \\		
		$\sigma^2$& The noise power\\
		$P_\text{F}\ (P_\text{B})$  &Transmit power of $\text{U}_\text{F}\ (\text{U}_\text{B})$\\
		{\color{black}$\rho_\text{F}\ (\rho_\text{B})$}  &{\color{black}Transmit signal-to-noise ratio (SNR) of $\text{U}_\text{F}\ (\text{U}_\text{B})$, $\rho_\text{F}=\frac{P_\text{F}}{\sigma^2}$, $\rho_\text{B}=\frac{P_\text{B}}{\sigma^2}$}\\	
		$P_m\ (\rho_m)$  &Maximal transmit power (SNR) of all users, $\rho_m=\frac{P_m}{\sigma^2}$\\
		$P_k\ (\rho_k)$  &Transmit power (SNR) of $\text{U}_k$\\		
		$\alpha_{\text{F}}\ (\alpha_{\text{B}})$ & Target channel gain of $\text{U}_{\text{F}}$ (U$_{\text{B}}$), $\alpha_{\text{F}}=\frac{\gamma_{\text{F}}}{\rho_\text{F}}$, $\alpha_{\text{B}}=\frac{\gamma_{\text{B}}}{\rho_\text{B}}$ \\
		{\color{black}$\alpha'_{\text{F}}\ (\alpha'_{\text{B}})$} & {\color{black}Minimal target channel gain of $\text{U}_{\text{F}}$ (U$_{\text{B}}$), $\alpha'_{\text{F}}=\frac{\gamma_{\text{F}}}{\rho_m}$, $\alpha'_{\text{B}}=\frac{\gamma_{\text{B}}}{\rho_m}$} \\
		$f_X(\cdot)$     & Probability density function of $X$  \\
		$F_X(\cdot)$ & Cumulative distribution function of $X$  \\
		$\mathbb{P}\{\cdot\}$ & Probability of an event \\
		$\triangleq$& Be defined as\\
		$f(a\mapsto b)$    & Replace $a$ in expression $f$ with $b$\\
		$\mathcal{CN}(\mu,\delta^2)$&Complex Gaussian random variable with mean $\mu$ and variance $\delta^2$ \\
		\hline
	\end{tabular}%
	\label{table1}%
\end{table}%

\subsection{Decoding with HSIC}
 {\color{black}In NOMA transmission, SIC is applied at the receiver to decode both users' signals, where one user's signal is decoded and subtracted from the compound received signal first, and then another user's signal is decoded. In fixed SIC (FSIC), the decoding order is pre-determined and is commonly employed in conventional SGF schemes \cite{2019_Ding_Semi-GF,2020_Yang_APA,2020_Jayanth_PIMRC,2020_Chao_Semi-GF,Zhang_2020_WCL}. However, FSIC may result in outage error floors, which leads to transmission robustness degradation of the GF users.

In order to effectively avoid outage error floors \cite{2020_Ding_HSIC}, HSIC is employed in this paper, where the decoding order is dynamically changed based on the instantaneous received SNR of $\text{U}_{\text{F}}$.}  Prior to user scheduling, the BS will first broadcast a threshold, denoted as $\tau_0\triangleq\text{max}\left\{0,\tau(|g|^2)\right\}$, to all the GF users \cite{2020_Ding_New_SGF}. Here, $\tau(|g|^2)=\alpha_\text{B}^{-1}|g|^2-1$, {\color{black}which} is derived based on the condition that the BS can successfully decode $\text{U}_{\text{B}}$'s signal at the first stage of SIC, namely, $\text{log}\left(1+\frac{\rho_\text{B}|g|^2}{\tau(|g|^2)+1}\right)\geq R_{\text{B}}$. {\color{black}Determing decoding order of HSIC can be divided into the following two cases:}
\begin{itemize}
	\item When the effective received SNR of $\text{U}_{\text{F}}$'s signal is larger than $\tau_0\ \text{(namely,}\ \rho_{\text{F}}|h|^2>\tau_0)$, the BS will decode $\text{U}_{\text{F}}$'s signal first with a data rate $\text{log}\left(1+\frac{\rho_{\text{F}}|h|^2}{\rho_{\text{B}}|g|^2+1}\right)$. Otherwise, with the opposite decoding order, $\text{U}_{\text{B}}$'s signal cannot be successfully decoded, since $\rho_{\text{F}}|h|^2>\tau_0$ leads to $\text{log}\left(1+\frac{\rho_{\text{B}}|g|^2}{\rho_{\text{F}}|h|^2+1}\right)<R_{\text{B}}$, {\color{black}while $\text{U}_{\text{B}}$'s QoS should be satisfied with priority in SGF transmissions.} Hence, $\text{U}_{\text{B}}$'s signal will be decoded at the second stage of SIC, with a data rate $\text{log}\left(1+{\color{black}\rho_{\text{B}}}|g|^2\right)$, which is the same as that achieved in OMA.
	\item {\color{black}When  $\rho_{\text{F}}|h|^2\leq\tau_0$, the BS will decode $\text{U}_{\text{B}}$'s signal first. Actually, in this case, for any decoding order, $\text{U}_{\text{B}}$ can achieve the same performance as that in OMA transmission, and the BS can decode $\text{U}_{\text{B}}$'s signal at either the first or the second stage of SIC. Accordingly, $\text{U}_{\text{F}}$ will achieve a data rate of $\text{log}\left(1+\rho_{\text{F}}|h|^2\right)$ or $\text{log}\left(1+\frac{\rho_{\text{F}}|h|^2}{\rho_{\text{B}}|g|^2+1}\right)$. Since $\text{log}\left(1+\rho_{\text{F}}|h|^2\right)>\text{log}\left(1+\frac{\rho_{\text{F}}|h|^2}{\rho_{\text{B}}|g|^2+1}\right)$, $\text{U}_{\text{B}}$'s signal will be decoded first to maximize $\text{U}_{\text{F}}$'s data rate.}
\end{itemize}

{\color{black}\subsection{Performance Metric}
In this paper, we use outage probability as a performance metric for different SGF schemes, which represents the probability {\color{black}that} the instantaneous achievable rate of an admitted GF user is less than a target rate. To gain more insights, diversity order will also be derived. The diversity order highlights the asymptotic scaling law of the outage probability to the transmit SNR, which is defined as \cite{diversity_order}
\begin{equation}\label{def_DV}
	\begin{aligned}
		d=-\underset{\rho\rightarrow\infty}{\text{lim}}\frac{\text{log}\mathcal{P}(\rho)}{\text{log}\rho},
	\end{aligned}
\end{equation}
where $\mathcal{P}(\rho)$ and $\rho$ denote the outage probability and transmit SNR, respectively.}

We would like the diversity order to be non-zero, which {\color{black}indicates that} the outage probability will {\color{black}constantly decrease with the increase of the transmit SNR}. Intuitively, we would also like the diversity order to be as large as possible, which means the outage probability will decrease fast {\color{black}with the increase of transmit SNR}. However, when $\rho\rightarrow\infty$, if the outage probability is a constant and {\color{black}is independent of} the transmit SNR, an outage error floor occurs, which may lead to a degradation of transmission robustness.

\section{SGF schemes}\label{two schemes}
In this section, we first introduce the BU-SGF scheme \cite{2020_Ding_New_SGF} with randomly deployed users, then the proposed fair CS-SGF scheme is presented.

\subsection{BU-SGF Scheme}
{\color{black} In BU-SGF scheme, the GF user which can achieve the maximal data rate will be admitted to access U$_\text{B}$’s channel. It can effectively utilize multi-user diversity and avoid the outage error floor of the admitted GF user, while guaranteeing the GB user's performance to be the same as it solely occupies the channel \cite{2020_Ding_New_SGF}. However, only the Rayleigh fading channel was taken into consideration in \cite{2020_Ding_New_SGF} and the analytical results can not be easily extended to other channel models. In reality, the users are always randomly distributed and their channels are also impacted by the path loss, therefore, it is necessary to consider these two factors in performance analysis. Instead, we extend the traditional BU-SGF scheme, where the admission procedure consists of the {\color{black}following} six steps:}
\begin{enumerate}
	\item The BS sends pilot signals.
	\item \textcolor{black}{Each user estimates its own CSI.}
	\item $\text{U}_\text{B}$ feedbacks its transmit SNR $\rho_\text{B}$, target rate $R_\text{B}$, and CSI $g$ to the BS.
	\item The BS calculates $\text{U}_\text{B}$'s decoding threshold $\tau_0$, and broadcasts $\text{U}_{\text{B}}$'s effective received SNR $\rho_\text{B}|g|^2$ and $\tau_0$ to all GF users.
	\item $\text{U}_k$ calculates its transmit data rate {\color{black}(if admitted), which is $\text{log}\left(1+\frac{\rho_k|h_k|^2}{\rho_{\text{B}}|g|^2+1}\right)$ if $\rho_k|h_k|^2>\tau_0$, or $\text{log}\left(1+\rho_k|h_k|^2\right)$ if $\rho_k|h_k|^2\leq \tau_0$,} {\color{black}where $\rho_k|h_k|^2$ denotes $\text{U}_k$'s effective received SNR (if admitted).} Note that, HSIC is applied here.
	\item The GF user with the maximal data rate will be admitted to transmit through distributed contention. Thus, the contention {\color{black}criterion} is each GF user's achievable data rate.
\end{enumerate}


In BU-SGF scheme, the achievable rate of $\text{U}_k$, {\color{black}$\text{log}\left(1+\frac{\rho_k|h_k|^2}{\rho_{\text{B}}|g|^2+1}\right)$ or $\text{log}\left(1+\rho_k|h_k|^2\right)$,} is an increasing function with respect to the channel gain $|h_k|^2$. {\color{black} By considering randomly distributed GF users  with different distances to the BS, the GF users closer to the BS can achieve higher data rate due to smaller path loss.  Since the BU-SGF scheme always admits the GF user with the largest data rate, the GF users closer to the BS will be scheduled more often. On the contrary, the GF users far from the BS will be scheduled seldomly and their generated data may become overflow, which results in severe fairness issue.} In the next subsection, we will propose a fair SGF scheme, which can schedule each GF user with equal probability.

	
\subsection{CS-SGF Scheme}

  We handle the fairness issue of SGF scheme by resorting to CDF-based scheduling, where the GF user with the largest CDF value about its channel gain, namely, the GF user whose channel is good enough relative to its own statistics, will be admitted \cite{2005_Park_packet_CDF}. Since all the GF users' channels are independent with each other and their channel gains' CDF values are uniformly distributed in [0, 1] \cite{2015-Jin-Fundamental_limits_CDF}, each GF user will have the same probability to obtain the largest CDF value and then the admission fairness can be guaranteed. 
 Assume that the CDF of $\text{U}_k$'s channel gain $|h_k|^2$ is denoted as $F_k(x)$\footnote{In this paper, we use ``CDF" to denote the cumulative distribution function, e.g., $F_k(x)$, and use ``CDF value" to represent the corresponding output value of a CDF with a specific input $x$.}. For Rayleigh small-scale fading, the CDF of $\text{U}_k$'s channel gain with a given distance $r_k$ (the distance from $\text{U}_k$ to the BS) can be expressed as
	\begin{equation}
	F_k(x|r_k)=1-e^{-(1+r_k^\alpha)x}.
	\end{equation}
	
{\color{black}We assume that the BS sends pilot signals at the beginning of each time slot for synchronizing uplink transmissions. In time division duplexing (TDD) mode, each GF user can estimate its channel {\color{black}gains} by measuring these pilot signals \cite{2017_Choi_ALOHA}. {\color{black}By applying non-parametric CDF based scheduling (NPCS) or parametric CDF based scheduling (PCS) \cite{2014_Nguyen_TSP_Leaning_methods}, each GF user can estimate its CDF based on the estimation of channel gains \cite{2015-Jin-Fundamental_limits_CDF}.} When a GF user wants to transmit data at a specific time slot, it will estimate the instantaneous channel gain at the beginning of that time slot, and then obtain the corresponding CDF value by substituting the estimated result into the CDF.} By using distributed contention control strategy  \cite{2006_Bletsas_network_path_selection}, the GF user with the largest CDF value can be admitted. The contention {\color{black}criterion} is each GF user's CDF value, namely, each GF user's backoff time is set to be inversely proportional to its CDF value. The admission procedure of CS-SGF scheme can be outlined as follows:
\begin{enumerate}
	{\color{black}\item The BS sends pilot signals.
	\item Each user estimates its own CSI (Based on that, each user can estimate its CDF after a long-term observation of the CSI.).
	\item $\text{U}_\text{B}$ feedbacks its transmit SNR $\rho_\text{B}$, target rate $R_\text{B}$, and CSI $g$ to the BS.
	\item The BS calculates $\text{U}_\text{B}$'s decoding threshold $\tau_0$, and broadcasts $\text{U}_{\text{B}}$'s effective received SNR $\rho_\text{B}|g|^2$ and $\tau_0$ to all GF users.
	\item The GF user with the maximal CDF value will be admitted to transmit through distributed contention, {\color{black}and }the contention {\color{black}criterion} is each GF user's CDF value.

\item The admitted GF user, U$_\text{F}$, calculates its transmit data rate, which is $\text{log}\left(1+\frac{\rho_\text{F}|h_\text{F}|^2}{\rho_{\text{B}}|g|^2+1}\right)$ if $\rho_\text{F}|h_\text{F}|^2>\tau_0$, or $\text{log}\left(1+\rho_\text{F}|h_\text{F}|^2\right)$ if $\rho_\text{F}|h_\text{F}|^2\leq \tau_0$. HSIC is applied here.
}
\end{enumerate}

\section{Performance Analysis for SGF Schemes with Fixed Transmit Power}\label{Performance of two schemes}
{\color{black}In this section, the outage probabilities and achieved diversity orders of $\text{U}_\text{F}$ are analyzed for both BU-SGF and CS-SGF schemes. Note that, since U$_{\text{B}}$ can always achieve the same performance as that in OMA for both schemes, we only characterize $\text{U}_{\text{F}}$'s outage performance \cite{2020_Ding_New_SGF}. The main steps of the analysis procedure are listed as follows:
\begin{enumerate}
	\item Derive $\text{U}_{\text{F}}$'s outage probability expressions. 
	
	\item Convert, combine, and/or simplify the outage probability expressions for easier calculation.
	
	\item Calculate the approximation expressions for the outage probabilities by using {\color{black}the CDF/pdf expressions of the users' channel gains.}
	
	\item Derive the high SNR approximations of the outage probabilities and the achieved diversity orders.
\end{enumerate}}

 For ease of theoretical analysis, the GF users are assumed to use a same fixed transmit SNR $\rho_\text{F}$ in both BU-SGF and CS-SGF schemes, namely, $\rho_k=\rho_\text{F}\ \text{for}\ k\in\{1,\dots,K\}$.




\subsection{Performance Analysis for BU-SGF Scheme}
{\color{black}According to the previous description, the outage probability of $\text{U}_{\text{F}}$ for BU-SGF scheme can be denoted as}
\begin{equation}
	\begin{aligned}
	\mathcal{P}_{\text{BU}}=&\sum_{k=0}^{K}\mathbb{P}\left\{E_k,\text{max}\{\text{log}(1+\rho_{\text{F}}|h_i|^2),1\leq i\leq k\}<R_{\text{F}},\right.\\ &\quad\qquad\text{max}\{\text{log}(1+\frac{\rho_{\text{F}}|h_i|^2}{\rho_{\text{B}}|g|^2+1}),k< i\leq K\}<R_{\text{F}}\},
	\end{aligned}
\end{equation}
where $E_k$ denotes the event that there are $k$ users' effective received SNRs are less than $\text{U}_{\text{B}}$'s decoding threshold $\tau_0$.

As the GF users transmit with fixed SNR $\rho_{\text{F}}$ and their channel gains are ordered as (\ref{channel order}), the outage probability can be rewritten as
\begin{equation}\label{P_BU_rewritten}
\begin{aligned}
\mathcal{P}_{\text{BU}}=&\mathbb{P}\left\{E_0,R_{K}^{\text{fs}}<R_{\text{F}}\right\}+\mathbb{P}\left\{E_K,R_{K}^{\text{ss}}<R_{\text{F}}\right\}+\sum_{k=1}^{K-1}\mathbb{P}\left\{E_k,R_{k}^{\text{ss}}<R_{\text{F}},R_{K}^{\text{fs}}<R_{\text{F}}\right\},\\
\end{aligned}
\end{equation}
{\color{black}where $R_{k}^{\text{fs}}\triangleq\text{log}(1+\frac{\rho_{\text{F}}|h_k|^2}{\rho_{\text{B}}|g|^2+1})$ and $R_{k}^{\text{ss}}\triangleq\text{log}(1+\rho_{\text{F}}|h_k|^2)$ represent the $k$-th ($k\in\{1,\dots,K\}$) GF user's achievable rates when its signal is decoded at the first and second stages of SIC, respectively. The superscripts ``fs'' and ``ss'' refer to the first and second stages of SIC, respectively.} In (\ref{P_BU_rewritten}), the first term shows the outage probability of the admitted GF user (namely, the $K$-th user whose signal is decoded at the first stage of SIC) when all GF users' effective received SNRs are larger than $\text{U}_\text{B}$'s decoding threshold $\tau_0$; the second term computes the outage probability of the admitted GF user (the $K$-th user whose signal is decoded at the second stage of SIC) when all the users' effective received SNRs are lower than $\tau_0$; and the third term considers the cases when there are $k\ (1\leq k\leq K-1)$ users with effective received SNRs lower than $\tau_0$, the probability of the admitted GF user (either the $k$-th user whose signal is decoded at the second stage of SIC, or the $K$-th user whose signal is decoded at the first stage of SIC) is in outage.

Note that, in the case of $|g|^2<\alpha_\text{B}=\frac{\gamma_{\text{B}}}{\rho_\text{B}}$, the decoding threshold $\tau_0\triangleq\text{max}\left\{0,\alpha_\text{B}^{-1}|g|^2-1\right\}=0$. Then the outage probability can be expressed as
\begin{equation}\label{The1_qian}
\begin{aligned}
\mathcal{P}_{\text{BU}}=
&\underbrace{\mathbb{P}\left\{|g|^2>\alpha_\text{B},E_0,R_{K}^{\text{fs}}<R_{\text{F}}\right\}}_{T_0}+\sum_{k=1}^{K-1}\underbrace{\mathbb{P}\left\{|g|^2>\alpha_\text{B},E_k,R_{k}^{\text{ss}}<R_{\text{F}},R_{K}^{\text{fs}}<R_{\text{F}}\right\}}_{T_k}\\
&+\underbrace{\mathbb{P}\left\{|g|^2>\alpha_\text{B},E_K,R_{K}^{\text{ss}}<R_{\text{F}}\right\}}_{T_K}+\underbrace{\mathbb{P}\left\{|g|^2<\alpha_\text{B},R_{K}^{\text{fs}}<R_{\text{F}}\right\}}_{T_{K+1}},
\end{aligned}
\end{equation}
where the terms $T_0$, $T_k\ (1\leq k\leq K-1)$, and $T_{K}$ denote the case $\tau_0>0$. Term $T_{K+1}$ represents the case $\tau_0=0$, thus all the GF users' signals should be decoded at the first stage of SIC and $\text{U}_K$ can achieve the maximal data rate in this case. An {\color{black}approximate} expression of $\mathcal{P}_{\text{BU}}$ is shown in the following theorem.

\begin{theorem}\label{theorem1}
Assume $K\geq2$. {\color{black}Depending on the value of $\gamma_{\text{B}}\gamma_{\text{F}}$, the outage probability of $\text{U}_{\text{F}}$ for BU-SGF scheme, $\mathcal{P}_{\text{BU}}$, takes two forms. Specifically,} 
{\color{black}
	\begin{subnumcases}{\label{The1_3} \mathcal{P}_{\text{BU}}=}
		\sum_{k=0}^{K}\bar{\eta}_kI_{1;k}
		+\sum_{k=0}^{K}\bar{\eta}_kI_{2;k}+I_{3}+I_{4},&$\text{if}\ \gamma_{\text{B}}\gamma_{\text{F}}<1$\\
		\sum_{k=0}^{K}\bar{\eta}_kI_{1;k}+\sum_{k=0}^{K}\bar{\eta}_kI_{2;k}(\alpha_2\mapsto\infty)+I_{3},&$\text{if}\ \gamma_{\text{B}}\gamma_{\text{F}}\geq 1$
	\end{subnumcases}
where,
\begin{equation}
	\begin{aligned}
		I_{1;k}=&\int_{\alpha_\text{B}}^{\alpha_1}f_\text{B}(w)\left[F_\text{F}\left(\frac{w}{\rho_{\text{F}}\alpha_\text{B}}-\frac{1}{\rho_\text{F}}\right)\right]^k\left[F_\text{F}\left(\alpha_\text{F}\rho_{\text{B}}w+\alpha_\text{F}\right)-F_\text{F}\left(\frac{w}{\rho_{\text{F}}\alpha_\text{B}}-\frac{1}{\rho_\text{F}}\right)\right]^{K-k}dw,
	\end{aligned}
\end{equation}

\begin{equation}
	\begin{aligned}
		I_{2;k}=&\int_{\alpha_1}^{\alpha_2}f_\text{B}(w)\left[F_\text{F}\left(\alpha_\text{F}\right)\right]^k\left[F_\text{F}\left(\alpha_\text{F}\rho_{\text{B}}w+\alpha_\text{F}\right)-F_\text{F}\left(\frac{w}{\rho_{\text{F}}\alpha_\text{B}}-\frac{1}{\rho_\text{F}}\right)\right]^{K-k}dw,
	\end{aligned}
\end{equation}

\begin{equation}
	\begin{aligned}
		I_3=&\int_{0}^{\alpha_\text{B}}f_\text{B}(w)\left[F_\text{F}(\alpha_\text{F}\rho_{\text{B}}w+\alpha_\text{F})\right]^Kdw,\\
	\end{aligned}	
\end{equation}

\begin{equation}
\begin{aligned}
I_{4}=&\left[1-F_\text{B}(\alpha_2)\right]\left[F_\text{F}\left(\alpha_\text{F}\right)\right]^K,
\end{aligned}
\end{equation}}
$\bar{\eta}_k=\frac{K!}{k!(K-k)!}$, $\alpha_1=\alpha_\text{B}(\gamma_{\text{F}}+1)$ and $\alpha_2=\frac{\alpha_\text{B}(\gamma_{\text{F}}+1)}{1-\gamma_{\text{B}}\gamma_{\text{F}}}$, {\color{black}$I_{2;k}(\alpha_2\mapsto\infty)$ denotes the expression of replacing $\alpha_2$ in $I_{2;k}$ with $\infty$.} $F_\text{F}(x)$, $F_\text{B}(y)$, and $f_\text{B}(y)$ are shown in (\ref{CDF_F_F}), (\ref{CDF_unordered}), and (\ref{pdf_GB}), respectively.

\end{theorem}
\begin{proof}
	Please refer to Appendix \ref{Proof of The1}.
\end{proof}

Following the similar steps for proofing Theorem \ref{theorem1}, the outage probability of the BU-SGF scheme in the case of only one GF user, namely, $K=1$, can be derived straightforwardly. We omit it for space limitation.

From the proof of Theorem \ref{theorem1}, we can find the main difference of the two cases (namely, $\gamma_{\text{B}}\gamma_{\text{F}}< 1$ and $\gamma_{\text{B}}\gamma_{\text{F}}\geq 1$) is that, an additional constraint of $|g|^2<\alpha_2$ is needed for the case of $\gamma_{\text{B}}\gamma_{\text{F}}< 1$, which effectively avoid the outage error floor of U$_\text{F}$.

{\color{black}\textit{{Remark 1:}} The outage probability expressions in Theorem \ref{theorem1} have favorable extensibility. We can obtain the outage probability of the BU-SGF scheme in \cite{2020_Ding_New_SGF} by substituting $F_\text{F}(x)=1-e^{-x}$ and $f_\text{B}(y)=e^{-y}$ into Theorem \ref{theorem1}. Moreover, they can also be applied to calculate the outage probability of  BU-SGF scheme with other fading channel models. For example, for the Nakagami-$m$ fading channel with fading parameter $m$ and variance $\lambda$, the outage probability of BU-SGF scheme can be obtained by substituting $F_\text{F}(x)=1-e^{-\frac{m}{\lambda }x}\sum_{s=0}^{m-1}\frac{(\frac{m}{\lambda}x)^s}{s!}$ and $f_\text{B}(y)=\frac{m^m}{\Gamma(m)\lambda^m}y^{m-1}e^{-\frac{my}{\lambda}}$ into Theorem \ref{theorem1}, where $\Gamma(\cdot)$ denotes the Gamma function.}

{\color{black}\textit{{Remark 2:}} It is quite complicated to derive closed form expressions for the outage probability of U$_\textnormal{F}$ in Theorem 1, mainly because $T_k$ ($0\leq k\leq K$) in (10) involve different order statistics. For example, $T_0$ is a function of three channel gains: $|g|^2$, $|h_1|^2$, and $|h_K|^2$, whereas $T_k$ ($1\leq k\leq K-2$) is a function of four channel gains: $|g|^2$, $|h_k|^2$, $|h_{k+1}|^2$, and $|h_K|^2$. In addition to these dependent order statistics $|h_k|^2$, $|h_{k+1}|^2$, and $|h_K|^2$, the complicated distribution functions in (2) and (3) make the derivation of closed forms more involved. However, at high SNR, insightful approximations can be obtained as shown in the following.}


\begin{corollary}\label{corollary1}
	 Assume $K\geq 2$ and $\rho_{\text{B}}=\rho_{\text{F}}\to\infty$. The high SNR approximation of $\mathcal{P}_{\text{BU}}$ (denoted as $\mathcal{\vec{P}}_{\text{BU}}$) can be expressed as
\begin{subnumcases}{\label{corollory 1} \vec{\mathcal{P}}_{\text{BU}}=}
\sum_{k=0}^{K}\bar{\eta}_k\vec{I}_{1;k}
+\sum_{k=0}^{K}\bar{\eta}_k\vec{I}_{2;k}+\vec{I}_{3}+\vec{I}_{4},& $\text{if}\ \gamma_{\text{B}}\gamma_{\text{F}}<1$\\
\sum_{k=0}^{K}\bar{\eta}_k\vec{I}_{1;k}+\sum_{k=0}^{K}\bar{\eta}_k\vec{H}_{2;k}+\vec{I}_{3},& $\text{if}\ \gamma_{\text{B}}\gamma_{\text{F}}\geq 1$
\end{subnumcases}
where,
\begin{equation}
\begin{aligned}
\vec{I}_{1;k}=&\frac{S_\text{B}S_{\text{F}}^K}{\rho_{\text{B}}^{K+1}}\sum_{i=0}^{K-k}\binom{K-k}{i}(\gamma_{\text{F}}+1)^{K-k-i}(\gamma_{\text{F}}-\gamma_{\text{B}}^{-1})^i\\
&\times\gamma_{\text{B}}^{i+1}\sum_{j=0}^{k}\binom{k}{j}(-1)^j\frac{(1+\gamma_{\text{F}})^{k-j+i+1}-1}{k-j+i+1},
\end{aligned}
\end{equation}
\begin{equation}
\begin{aligned}
\vec{I}_{2;k}=&\frac{S_\text{B}S_{\text{F}}^K}{\rho_{\text{B}}^{K+1}}\sum_{i=0}^{K-k}\binom{K-k}{i}(\gamma_{\text{F}}+1)^{K-k-i}(\gamma_{\text{F}}-\gamma_{\text{B}}^{-1})^i\gamma_{\text{F}}^k\frac{\tilde{\alpha}_2^{i+1}-\tilde{\alpha}_1^{i+1}}{i+1},
\end{aligned}
\end{equation}
\begin{equation}
\begin{aligned}
\Omega=&\left(-\frac{1}{2}\right)^{K-k}\frac{1}{D_0+D_1}\sum_{n=1}^{N}\Phi_nc_n\sum_{m=0}^{K-k}\binom{K-k}{m}(-1)^m\\
&\times\sum_{\sum_{l=0}^{L}p_l={K-k-m}}\binom{{K-k-m}}{p_0,\dots,p_L}\sum_{\sum_{l=0}^{L}q_l=m}\binom{m}{q_0,\dots,q_L}\left(\prod_{l=0}^{L}\Psi_l^{p_l+q_l}\right),\\
\end{aligned}
\end{equation}
$\vec{I}_{3}=\frac{S_\text{B}S_{\text{F}}^K\gamma_{\text{F}}^{K}}{\rho_{\text{B}}^{K+1}(K+1)}\left[(1+\gamma_{\text{B}})^{K+1}-1\right]$, $\vec{I}_{4}=(\frac{S_{\text{F}}\gamma_{\text{F}}}{\rho_{\text{B}}})^K\left(1-\frac{S_\text{B}\tilde{\alpha}_{2}}{\rho_{\text{B}}}\right)$, $S_\text{B}=\frac{1}{D_0+D_1}\sum_{n=1}^{N}\Phi_nc_n$, $S_\text{F}=\frac{1}{2}\sum_{l=1}^{L}\Psi_l\mu_l$, $\tilde{\alpha}_1=\gamma_{\text{B}}(\gamma_{\text{F}}+1)$, $\tilde{\alpha}_2=\frac{\gamma_{\text{B}}(\gamma_{\text{F}}+1)}{1-\gamma_{\text{B}}\gamma_{\text{F}}}$, $\vec{H}_{2;k}=
\frac{\Omega\left(S_\text{F}\gamma_\text{F}\right)^k}{\rho_\text{B}^k}(\sum_{l=0}^{L}\left(p_l\mu_l\gamma_\text{F}+q_l\mu_l\gamma_\text{B}^{-1}\right)+c_n)^{-1}$.
\end{corollary}
\begin{proof}
	Please refer to Appendix \ref{Proof of The2}.
\end{proof}

{\color{black}It looks very complicated at the first sight of (\ref{corollory 1}), while most of the variables refer to the given system parameters. We only need to observe the exponent of $\rho_{\text{B}}$ in each term, which determines the achieved diversity order.} Comparing the terms in (\ref{corollory 1}), we have the following observations in high SNR region (namely, $\rho_{\text{B}}=\rho_{\text{F}}\to\infty$). In the case of $\gamma_{\text{B}}\gamma_{\text{F}}<1$, $\vec{I}_{4}$ is inversely proportional to $\rho_{\text{B}}^{K}$, and all other terms are inversely proportional to $\rho_{\text{B}}^{K+1}$. {\color{black}By applying (\ref{def_DV}), we know that $\text{U}_{\text{F}}$ can achieve a diversity order of $K$ (i.e., a full diversity order given $K$ GF users). Moreover, the outage probability of the scheduled GF user decreases with the increasing of the number of GF users $K$, and also decreases with the decreasing of the GF user's target SINR $\gamma_{\text{F}}$ and its distribution region $D_{\text{F}}$ (since $S_\text{F}$ is an increasing function of $D_{\text{F}}$). However, in the case of $\gamma_{\text{B}}\gamma_{\text{F}}\geq1$, $\vec{H}_{2;0}$ is a constant, which is irrelevant to $\rho_{\text{B}}$. Hence, $\text{U}_{\text{F}}$ achieve a diversity order of 0. In addition to $K$, $\gamma_{\text{F}}$, and $D_{\text{F}}$, the outage probability in this case is also related with the GB user's target SINR $\gamma_{\text{B}}$ and distribution region $D_0$ and $D_1$.}



\subsection{Performance Analysis for CS-SGF Scheme}
Recall that for the CS-SGF scheme, the GF user with the maximal CDF value will be admitted, whose instantaneous channel gain is denoted as $|h|^2$. Based on the descriptions in Sections \ref{system model} and \ref{two schemes}, the outage probability of $\text{U}_\text{F}$ can be expressed as
\begin{equation}\label{OP_GF_1}
\begin{aligned}
\mathcal{P}_{\text{CS}}
=&\underbrace{\mathbb{P}\left\{|g|^2<\alpha_\text{B},\frac{\rho_\text{F}|h|^2}{\rho_{\text{B}}|g|^2+1}<\gamma_{\text{F}}\right\}}_{\Delta_1}\\
&+\underbrace{\mathbb{P}\left\{|g|^2>\alpha_\text{B},\rho_\text{F}|h|^2>\tau(|g|^2),\frac{\rho_\text{F}|h|^2}{\rho_{\text{B}}|g|^2+1}<\gamma_{\text{F}}\right\}}_{\Delta_2}\\
&+\underbrace{\mathbb{P}\left\{|g|^2>\alpha_\text{B},\rho_\text{F}|h|^2<\tau(|g|^2),\rho_\text{F}|h|^2<\gamma_{\text{F}}\right\}}_{\Delta_3},\\
\end{aligned}
\end{equation}
where $\Delta_1$ denotes the case $\tau_0=0$, and $\text{U}_{\text{F}}$'s signal is decoded at the first stage of SIC. $\Delta_2$ ($\Delta_3$) denotes the case when $\tau_0>0$, and $\text{U}_{\text{F}}$'s signal is decoded at the first (second) stage of SIC.

\begin{theorem}\label{theorem2}
	{\color{black}Depending on the value of $\gamma_{\text{B}}\gamma_{\text{F}}$, the outage probability of $\text{U}_{\text{F}}$ for CS-SGF scheme, $\mathcal{P}_{\text{CS}}$, can be approximated as two forms, namely,}
\begin{subnumcases}{\label{OP_theorem21} \mathcal{P}_{\text{CS}}\approx}
\frac{\Xi_1c_n}{\Theta_1}e^{-k\mu_l\alpha_\text{F}}(1-e^{-\Theta_1\alpha_2})+\frac{\Xi_1c_n}{\Theta_2}e^{\frac{k\mu_l}{\rho_{\text{F}}}}\left(e^{-\Theta_2\alpha_2}-e^{-\Theta_2\alpha_1}\right) \notag \\
+\Xi_2e^{-c_n\alpha_1}(1-e^{-\mu_l\alpha_\text{F}})^K,&\text{if}\ $\gamma_{\text{B}}\gamma_{\text{F}}<1$\\
\frac{\Xi_1c_n}{\Theta_1}e^{-k\mu_l\alpha_\text{F}}-\frac{\Xi_1c_n}{\Theta_2}e^{\frac{k\mu_l}{\rho_{\text{F}}}}e^{-\Theta_2\alpha_1}+\Xi_2e^{-c_n\alpha_1}(1-e^{-\mu_l\alpha_\text{F}})^K,&\text{if}\ $\gamma_{\text{B}}\gamma_{\text{F}}\geq1$
\end{subnumcases}
where $\Xi_1=\frac{1}{2(D_1+D_0)}\sum_{l=1}^{L}\Psi_l\sum_{k=0}^{K}\binom{K}{k}(-1)^k\sum_{n=1}^{N}\Phi_n$, $\Xi_2=\frac{1}{2(D_1+D_0)}\sum_{l=1}^{L}\Psi_l\sum_{n=1}^{N}\Phi_n$, $\Theta_1=k\mu_l\rho_{\text{B}}\alpha_\text{F}+c_n$, and $\Theta_2=\frac{k\mu_l}{\rho_{\text{F}}\alpha_\text{B}}+c_n$.
\end{theorem}

\begin{proof}
	Please refer to Appendix \ref{Proof SC-SGF}.
\end{proof}

To obtain more insights, we derive the high SNR asymptotic expressions of the outage probabilities as well.

\begin{corollary}\label{corollary3}
When $\rho_{\text{B}}=\rho_{\text{F}}\rightarrow\infty$, if $\gamma_{\text{B}}\gamma_{\text{F}}<1$, the high SNR approximation of $\mathcal{P}_{\text{CS}}$ can be expressed as
\begin{equation}\label{corollory3_1}
\begin{aligned}
\vec{\mathcal{P}}_{\text{CS}}=&
\frac{\Xi_2c_n}{\rho_{\text{B}}}\left(\frac{\mu_l\gamma_{\text{F}}}{\rho_{B}}\right)^K\sum_{k=0}^{K}\binom{K}{k}\frac{\tilde{\alpha}_2^{k+1}}{k+1}+\Xi_2\left(\frac{\mu_l\gamma_{\text{F}}}{\rho_{\text{B}}}\right)^K\\
&+\frac{\Xi_2c_n\mu_l^K}{\rho_{\text{B}}^{K+1}}\sum_{k=0}^{K}\binom{K}{k}(-1)^{K-k}\frac{\tilde{\alpha}_1^{k+1}-\tilde{\alpha}_2^{k+1}}{\gamma_{\text{B}}^k(k+1)},\\
\end{aligned}
\end{equation}
If $\gamma_{\text{B}}\gamma_{\text{F}}\geq 1$, the high SNR approximation of $\mathcal{P}_{\text{CS}}$ equals
\begin{equation}\label{coroll33_2}
\begin{aligned}
\vec{\mathcal{P}}_{\text{CS}}=
&\Xi_1c_n\left({\color{black}\Theta_1^{'-1}-\Theta_2^{'-1}}\right)+\frac{\Xi_2c_n\mu_l^K}{\rho_{\text{B}}^{K+1}}\sum_{k=0}^{K}\binom{K}{k}(-1)^{K-k}\frac{\gamma_{\text{B}}(1+\gamma_{\text{F}})^{k+1}-\gamma_{\text{B}}}{k+1}\\
&+\frac{\Xi_2c_n}{\rho_{\text{B}}}\left(\frac{\mu_l\gamma_{\text{F}}}{\rho_{\text{B}}}\right)^K\sum_{k=0}^{K}\binom{K}{k}\frac{\gamma_{\text{B}}^{k+1}}{k+1}+\Xi_2\left(\frac{\mu_l\gamma_{\text{F}}}{\rho_{\text{B}}}\right)^K,
\end{aligned}
\end{equation}
{\color{black}where $\Theta'_1=k\mu_l\gamma_\text{F}+c_n$ and $\Theta'_2=\frac{k\mu_l}{\gamma_\text{B}}+c_n$.}
\end{corollary}
\begin{proof}
	Please refer to Appendix \ref{Proof of corollary3}.
\end{proof}


It can be observed that, the second term of (\ref{corollory3_1}) is inversely proportional to $\rho_{\text{B}}^K$, and the other two terms of (\ref{corollory3_1}) are inversely proportional to $\rho_{\text{B}}^{K+1}$. {\color{black}Hence, $\text{U}_{\text{F}}$ can achieve a {\color{black}full} diversity order of $K$.} Moreover, it can also be observed that the first term of (\ref{coroll33_2}) is a constant,   {\color{black}thus $\text{U}_{\text{F}}$ can achieve a diversity order of $0$.}


We know from {\color{black}Corollaries \ref{corollary1} and \ref{corollary3}} that, with fixed transmit power, both BU-SGF and CS-SGF schemes can avoid outage error floors only in the case of $\gamma_{\text{B}}\gamma_{\text{F}}<1$, which restricts the application scenarios for realizing robust transmissions. In order to eliminate the outage error floors for the case of $\gamma_{\text{B}}\gamma_{\text{F}}\geq1$ as well, in the following, we propose a distributed power control strategy for the GB and GF users.

\section{Power Control Strategy and Associated Performance Analysis}\label{Propose PC and performance analyze}

In this section, the proposed power control strategy is presented first. Then, the outage performances of both BU-SGF and CS-SGF schemes applied with the power control strategy are analyzed.

\subsection{Proposed Power Control Strategy}	

\subsubsection{Power Control for {\color{black}the} GF Users}
In both BU-SGF and CS-SGF schemes, the GF users always transmit with fixed SNR $\rho_{\text{F}}$. If GF user U$_k$'s effective received SNR $\rho_{\text{F}}|h_k|^2$ is larger than the GB user $\text{U}_\text{B}$'s decoding threshold $\tau_0$, U$_k$'s signal should be decoded at the first stage of SIC with a data rate $\text{log}\left(1+\frac{\rho_{\text{F}}|h_k|^2}{\rho_{\text{B}}|g|^2+1}\right)$. On the other hand, if $\rho_{\text{F}}|h_k|^2\leq\tau_0$, U$_k$'s signal will be decoded at the second stage of SIC, with a data rate $\text{log}\left(1+\rho_{\text{F}}|h_k|^2\right)$. {\color{black}For the two cases, U$_k$'s rate can be explicitly expressed as
\begin{equation}\label{motivation_PC}
	R_k=
	\begin{cases}
		\text{log}\left(1+\frac{\rho_{\text{F}}|h_k|^2}{\rho_{\text{B}}|g|^2+1}\right),&\rho_{\text{F}}|h_k|^2>\tau_0\\
		\text{log}\left(1+\rho_{\text{F}}|h_k|^2\right),&\rho_{\text{F}}|h_k|^2\leq\tau_0\\
	\end{cases}.
\end{equation}}
\ It can be observed from (\ref{motivation_PC}) that, in the context of $\rho_{\text{F}}|h_k|^2>\tau_0$ and $\tau_0>0$, U$_k$ can reduce its transmit SNR $\rho_{\text{F}}$ to be less than $\frac{\tau_0}{|h_k|^2}$ to make its signal be decoded at the second stage of SIC, and the achievable rate will be changed from $\text{log}\left(1+\frac{\rho_{\text{F}}|h_k|^2}{\rho_{\text{B}}|g|^2+1}\right)$ to $\text{log}\left(1+\rho'_{\text{F}}|h_k|^2\right)$ accordingly, where  $\rho'_{\text{F}}\leq\frac{\tau_0}{|h_k|^2}$. On the other hand, if $\rho_{\text{F}}|h_k|^2<\tau_0$ and $\rho_m|h_k|^2>\tau_0$, U$_k$ can increase its transmit SNR and transform its rate from $\text{log}\left(1+\rho_{\text{F}}|h_k|^2\right)$ to $\text{log}\left(1+\frac{\rho''_{\text{F}}|h_k|^2}{\rho_{\text{B}}|g|^2+1}\right)$, where $\frac{\tau_0}{|h_k|^2}<\rho''_{\text{F}}\leq \rho_m$. This is the motivation of our power control strategy for the GF users.



The power control strategy for the GF user aims at maximizing its data rate. Hence, $\rho_k$ should be set as $\rho_m$, since both $\text{log}\left(1+\frac{\rho_k|h_k|^2}{\rho_{\text{B}}|g|^2+1}\right)$ and $\text{log}\left(1+\rho_k|h_k|^2\right)$ are \textcolor{black}{monotonically increasing} functions of $\rho_k$. Similar with (\ref{motivation_PC}), when $\text{U}_k$'s maximal effective received SNR is larger than U$_\text{B}$'s decoding threshold, namely, $\rho_m|h_k|^2>\tau_0$, $\text{U}_k$ may tune its transmit SNR based on the values of $\text{log}\left(1+\frac{\rho_m|h_k|^2}{\rho_\text{B}|g|^2+1}\right)$ and $\text{log}\left(1+\tau_0\right)$. To be specific, in the case of $\rho_m|h_k|^2>\tau_0$, U$_k$'s transmit SNR may be set as $\rho_m$ or $\frac{\tau_0}{|h_k|^2}$ in order to maximize U$_k$'s data rate, and the achievable rate is $\text{log}\left(1+\frac{\rho_m|h_k|^2}{\rho_\text{B}|g|^2+1}\right)$ or $\text{log}\left(1+\tau_0\right)$, respectively. Thus the power control strategy for U$_k$ is
\begin{equation}\label{power control strategy}
\rho_k=
\begin{cases}
\frac{\tau_0}{|h_k|^2},&\text{if}\ \frac{\tau_0}{\rho_m}<|h_k|^2<\frac{\tau_0(1+\rho_{\text{B}}|g|^2)}{\rho_m}\\
\rho_m,&\text{otherwise}\\
\end{cases}.
\end{equation}

{\color{black}Note that, recently, a similar power control strategy for HSIC was proposed in \cite{2021_Sun_A_new} for cognitive radio-inspired NOMA system, which however only considered one secondary user with Rayleigh fading channel.}

\subsubsection{Power Control for the GB User}
{\color{black}The power control strategy for the GB user is to increase the probability of decoding the GF user's signal at the second stage of SIC, by which a higher energy efficiency can be achieved. Energy efficiency is quite important for the GF users, which are energy constrained in many IoT applications \cite{Shirvanimoghaddam_2017_massive NOMA,2017_Dawy_ Toward mMTC}. First, if U$_\textnormal{B}$'s channel gain is too low to support its target rate even with the maximal transmit power (namely, $\text{log}\left(1+\rho_m|g|^2\right)<R_\textnormal{B}$), U$_\text{B}$'s transmit SNR $\rho_\textnormal{B}$ is set as $0$. {\color{black}Since the transmission of U$_\text{B}$ will be failed definitely in this case}, and the transmit power of U$_\textnormal{B}$ will cause negative impact on the decoding of U$_\textnormal{F}$'s signal. 
	

For the case of $\text{log}\left(1+\rho_m|g|^2\right)\geq R_\textnormal{B}$, in order to increase the probability of decoding the GF user's signal at the stage of SIC, the transmit SNR of the GB user is set as $\rho_m$. Because the GF user's signal can be decoded at the second stage of SIC when $\rho_m|h_k|^2<\tau_0(1+\rho_\text{B}|g|^2)$ according to (\ref{power control strategy}). Obviously, $\tau_0(1+\rho_\text{B}|g|^2)$ is monotonically increasing with $\rho_\text{B}$. Hence, the power control strategy for U$_\text{B}$ is
\begin{equation}\label{power control strategy-UB}
	\rho_\textnormal{B}=
	\begin{cases}
		0,&\text{if}\ |g|^2<\alpha'_\textnormal{B}\\
		\rho_m,&\text{otherwise}\\
	\end{cases}.
\end{equation}}
\ {\textit{Remark 3:}} Compared to fixed transmit power strategy in BU-SGF and CS-SGF schemes, the proposed power control strategy does not introduce extra signaling overhead and can be executed distributedly. We can see from (\ref{power control strategy}) that $\rho_k$ is decided by the maximal transmit SNR $\rho_m$, $|h_k|^2$, $\tau_0$, and $\rho_{\text{B}}|g|^2$. Note that, $\rho_m$ and $|h_k|^2$ are already known by $\text{U}_k$, while $\tau_0$ and $\rho_{\text{B}}|g|^2$ are also needed to be broadcasted by the BS for BU-SGF and CS-SGF schemes.



In the following, the outage performance of both BU-SGF scheme with power control strategy (BU-SGF-PC) and CS-SGF scheme with power control strategy (CS-SGF-PC) will be analyzed.

\subsection{Performance Analysis for BU-SGF-PC Scheme}

Following the same steps of deriving (\ref{The1_qian}), the admitted GF user's outage probability for BU-SGF-PC scheme can be calculated as
\begin{equation}\label{Th3_1_PC}
{\color{black}\begin{aligned}
\mathcal{P}_{\text{BU}}^{\text{PC}}=&
\underbrace{\mathbb{P}\left\{|g|^2>\alpha'_\text{B},E_0^{\text{PC}},\gamma_{K}^{\text{fs}}>\tau'(|g|^2),\gamma_{K}^{\text{fs}}<\gamma_{\text{F}}\right\}}_{T_{0,1}^{\text{PC}}}\\
&\underbrace{+\mathbb{P}\left\{|g|^2>\alpha'_\text{B},E_0^{\text{PC}},\gamma_{K}^{\text{fs}}<\tau'(|g|^2),\tau'(|g|^2)<\gamma_{\text{F}}\right\}}_{T_{0,2}^{\text{PC}}}\\
&+\sum_{k=1}^{K-1}\underbrace{\mathbb{P}\left\{|g|^2>\alpha'_\text{B},E_k^{\text{PC}},\gamma_{k}^{\text{ss}}<\gamma_{\text{F}},\gamma_{K}^{\text{fs}}>\tau'(|g|^2),\gamma_{K}^{\text{fs}}<\gamma_{\text{F}}\right\}}_{T_{k,1}^{\text{PC}}}\\
&+\sum_{k=1}^{K-1}\underbrace{\mathbb{P}\left\{|g|^2>\alpha'_\text{B},E_k^{\text{PC}},\gamma_{k}^{\text{ss}}<\gamma_{\text{F}},\gamma_{K}^{\text{fs}}<\tau'(|g|^2),\tau'(|g|^2)<\gamma_{\text{F}}\right\}}_{T_{k,2}^{\text{PC}}}\\
&+\underbrace{\mathbb{P}\left\{|g|^2>\alpha'_\text{B},E_K^{\text{PC}},\gamma_{K}^{\text{ss}}<\gamma_{\text{F}}\right\}}_{T_K^{\text{PC}}}
+\underbrace{\mathbb{P}\left\{|g|^2<\alpha'_\text{B},\gamma_{K}^{\text{ss}}<\gamma_{\text{F}}\right\}}_{T_{K+1}^{\text{PC}}},
\end{aligned}}
\end{equation}
where $\gamma_{k}^{\text{fs}}=\frac{\rho_m|h_k|^2}{\rho_{m}|g|^2+1}$, $\gamma_{k}^{\text{ss}}=\rho_m|h_k|^2$, {\color{black}and $\tau'(|g|^2)=\frac{|g|^2}{\alpha'_{\text{B}}}-1$}. $E_k^{\text{PC}}\ (1\leq k\leq K)$ denotes the event that there are $k$ users' maximal received SNRs are less than $\tau_0$. $T_{K+1}^{\text{PC}}$ denotes the outage probability of the admitted GF user U$_K$ when $\tau_0=0$, while the other terms represent the case $\tau_0>0$. More specifically, $T_{0}^{\text{PC}}\triangleq T_{0,1}^{\text{PC}}+T_{0,2}^{\text{PC}}$ denotes, when all GF users' maximal effective received SNRs are larger than $\tau_0$, the outage probability of the admitted GF user U$_K$, whose signal is transmitted with SNR $\rho_m$ if $\gamma_{K}^{\text{fs}}>\tau'(|g|^2)$, or with SNR $\frac{\tau'(|g|^2)}{|h_K|^2}$ if $\gamma_{K}^{\text{fs}}<\tau'(|g|^2)$.
Similarly, $T_{k}^{\text{PC}}\triangleq T_{k,1}^{\text{PC}}+T_{k,2}^{\text{PC}}$ represents, when there are $k\ (1\leq k\leq K-1)$ users' maximal effective received SNRs are larger than $\tau_0$, the outage probability of the admitted GF user (U$_k$ or U$_K$). Here, U$_K$'s  transmit SNR is $\rho_m$ if $\gamma_{K}^{\text{fs}}>\tau'(|g|^2)$, or $\frac{\tau'(|g|^2)}{|h_K|^2}$ if $\gamma_{K}^{\text{fs}}<\tau'(|g|^2)$. $T_{K}^{\text{PC}}$ shows, when all the $K$ GF users maximal effective received SNRs are less than $\tau_0$, the outage probability of the admitted GF user (U$_K$ in this case).

\begin{theorem}\label{theorem3}
The outage probability of $\text{U}_{\text{F}}$ for BU-SGF-PC scheme, $\mathcal{P}_{\text{BU}}^{\text{PC}}$, can be expressed as 
\begin{equation}\label{Theorem3}
	\begin{aligned}
	{\color{black}\mathcal{P}_{\text{BU}}^{\text{PC}}=\sum_{k=0}^{K}\bar{\eta}_kI_{1;k}(\rho_\text{F}\mapsto \rho_m,\rho_\text{B}\mapsto \rho_m)+\left[1-F_\text{B}(\alpha'_1)\right]\left[F_\text{F}\left(\alpha'_{\text{F}}\right)\right]^K+F_\text{B}(\alpha'_\text{B})[F_\text{F}(\alpha'_\text{F})]^K,}
	\end{aligned}
\end{equation}
{\color{black}where  $\alpha'_1=\alpha'_\text{B}(\gamma_{\text{F}}+1)$.}
\end{theorem}


\begin{proof}
Please refer to Appendix \ref{Proof of The3}.
\end{proof}


{\color{black}By following the same steps of deriving (16), the high SNR approximation of $\mathcal{P}_{\text{BU}}^{\text{PC}}$  when $\rho_m\to\infty$ can be expressed as}
\begin{equation}\label{The4_1}	
	\begin{aligned}
	{\color{black}\vec{\mathcal{P}}_{\text{BU}}^{\text{PC}}=}
	&{\color{black}\sum_{k=0}^{K}\bar{\eta}_k\vec{I}_{1;k}(\rho_\text{B}\mapsto \rho_m)+\left(\frac{S_{\text{F}}\gamma_{\text{F}}}{\rho_{m}}\right)^K\left(1-\frac{S_\text{B}\tilde{\alpha}_{1}}{\rho_{m}}\right)+\left(\frac{S_{\text{F}}\gamma_{\text{F}}}{\rho_{m}}\right)^K\left(\frac{S_\text{B}\gamma_{\text{B}}}{\rho_{m}}\right).}\\
	\end{aligned}
	\end{equation}
{\color{black}And $\text{U}_{\text{F}}$ can achieve a full diversity order of $K$.}



\subsection{Performance Analysis for CS-SGF-PC Scheme}\label{sec_vc}
Similar with (\ref{OP_GF_1}), the outage probability of $\text{U}_\text{F}$ for CS-SGF-PC scheme can be formulated as
\begin{equation}\label{OP_PC_1}
\begin{aligned}
{\color{black}\mathcal{P}_{\text{CS}}^{\text{PC}}}
&{\color{black}=\underbrace{\mathbb{P}\{|g|^2<\alpha'_{\text{B}},\rho_{m}|h|^2<\gamma_{\text{F}}\}}_{\Delta_6}+\underbrace{\mathbb{P}\{|g|^2>\alpha'_{\text{B}},\rho_{m}|h|^2<\tau'(|g|^2),\rho_{m}|h|^2<\gamma_{\text{F}}\}}_{\Delta_3}}\\
&{\color{black}+\underbrace{\mathbb{P}\{|g|^2>\alpha'_{\text{B}},\rho_{m}|h|^2>\tau'(|g|^2),\frac{\rho_{m}|h|^2}{\rho_{m}|g|^2+1}>\tau'(|g|^2),\frac{\rho_{m}|h|^2}{\rho_{m}|g|^2+1}<\gamma_{\text{F}}\}}_{\Delta_4}}\\
&{\color{black}+\underbrace{\mathbb{P}\{|g|^2>\alpha'_{\text{B}},\rho_{m}|h|^2>\tau'(|g|^2),\frac{\rho_{m}|h|^2}{\rho_{m}|g|^2+1}<\tau'(|g|^2),\tau'(|g|^2)<\gamma_{\text{F}}\}}_{\Delta_5}},\\
\end{aligned}
\end{equation}
 where $\Delta_{4}$ and $\Delta_{5}$ represent, when $\rho_m|h|^2>\tau'(|g|^2)$, $\text{U}_\text{F}$'s signal is transmitted with SNRs $\rho_m$ and $\frac{\tau'(|g|^2)}{|h|^2}$ (here $\tau_0=\tau'(|g|^2)$, since {\color{black}$|g|^2>\alpha'_{\text{B}}$} in $\Delta_{4}$ and $\Delta_{5}$), respectively.

\begin{theorem}\label{theorem4}
The outage probability of $\text{U}_{\text{F}}$ for CS-SGF-PC scheme can be approximated as
\begin{equation}\label{OP_theorem2}
\begin{aligned}
{\color{black}\mathcal{P}_{\text{CS}}^{\text{PC}}\approx}
&{\color{black}\frac{\Xi_1c_n}{\Theta'_1}e^{-k\mu_l\alpha'_{\text{F}}}(e^{-\Theta'_1\alpha'_{\text{B}}}-e^{-\Theta'_1\alpha'_{1}})+\Xi_2(1-e^{-c_n\alpha'_{\text{B}}})(1-e^{-\mu_l\alpha'_{\text{F}}})^K}\\
&{\color{black}+\Xi_2e^{-c_n\alpha'_{1}}(1-e^{-\mu_l\alpha'_{\text{F}}})^K.}
\end{aligned}
\end{equation}
\end{theorem}

\begin{proof}
	Please refer to Appendix \ref{Proof of The4}.
\end{proof}


Although Theorems \ref{theorem2} and \ref{theorem4} are derived for CS-SGF and CS-SGF-PC schemes, in the {\color{black}special} case of $K=1$, these expressions can also be used to evaluate the performance of random selection SGF scheme with HSIC decoding, which can also achieve fair access for the GF users. That is because all the GF users are randomly distributed within $\mathcal{D}_\text{F}$ with the same distribution, and these distributions are independent with each other.

{\color{black}Following the similar steps of deriving (\ref{corollory3_1}), the high SNR approximation of $\mathcal{P}_{\text{CS}}^{\text{PC}}$ is}
\begin{equation}\label{lemma2_2}
	\begin{aligned}
		\vec{\mathcal{P}}_{\text{CS}}^{\text{PC}}=
		&{\color{black}\frac{\Xi_2c_n}{\rho_{m}}\left(\frac{\mu_l\gamma_{\text{F}}}{\rho_m}\right)^K\sum_{k=0}^{K}\binom{K}{k}\frac{(\gamma_{\text{B}}+\gamma_{\text{B}}\gamma_{\text{F}})^{k+1}-\gamma_{\text{B}}^{K+1}}{k+1}}\\
		&{\color{black}+\Xi_2\left(\frac{c_n\gamma_{\text{B}}}{\rho_m}\right)\left(\frac{\mu_l\gamma_{\text{F}}}{\rho_m}\right)^K+\Xi_2\left(\frac{\mu_l\gamma_{\text{F}}}{\rho_m}\right)^K.}\\
	\end{aligned}
\end{equation}
{\color{black}And $\text{U}_{\text{F}}$ can achieve a {\color{black}full} diversity order of $K$.}

{\color{black}\textit{Remark 4:} From the derivation process of Theorem \ref{theorem1} and Corollaries \ref{corollary1} (Theorem \ref{theorem2} and Corollary \ref{corollary3}), we can see that,  for BU-SGF (CS-SGF) scheme and in the case of $\gamma_{\text{B}}\gamma_{\text{F}}\geq1$, the outage error floor results from $H_{2;0}$ in (16b) ($\Delta_{2}$ in (\ref{OP_GF_1})). More specifically, the unbounded U$_\text{B}$'s channel gain $|g|^2>\alpha_{1}$ ($|g|^2>\alpha_{\text{B}}$) leads to the outage error floors for BU-SGF (CS-SGF) scheme. However, we can observe from the derivation process of Theorem \ref{theorem3} (Theorem \ref{theorem4}) that, an additional constraint of $|g|^2<\alpha'_1$ is introduced after applying the power control strategy, which effectively eliminates the error floors.}






{\color{black}\textit{Remark 5:} Comparing {\color{black}Corollary \ref{corollary1} with (\ref{The4_1}) and Corollary \ref{corollary3} with (\ref{lemma2_2})}, we can obtain an interesting insight that, when $\gamma_{\text{B}}\gamma_{\text{F}}\geq1$ the proposed power control strategy can effectively avoid the outage error floor and greatly improve the outage performance. But when $\gamma_{\text{B}}\gamma_{\text{F}}<1$, the BU-SGF (CS-SGF) scheme can achieve the same outage performance as BU-SGF-PC (CS-SGF-PC) scheme in high SNR region. In other words,  in the case of $\gamma_{\text{B}}\gamma_{\text{F}}<1$, the power control strategy can only improve the outage performance of BU-SGF (CS-SGF) scheme at the moderate SNR region, but it can not improve the outage performance at high SNR region, which will be demonstrated in the simulation results.}

	\section{Simulation Results and Discussions}\label{Simulation and discussion}
	

	In this section, the performances of the SGF transmission schemes are compared and the accuracy of the theoretical analyses {\color{black}is} examined through computer simulations. In existing studies, only \cite{2020_Chao_Semi-GF,Zhang_2020_WCL} investigated the impact of user locations on SGF schemes. Thus, the random selection SGF scheme with fixed SIC orders (termed as RS-SGF-FSIC) in \cite{2020_Chao_Semi-GF,Zhang_2020_WCL} is used as {\color{black}a} benchmark, where the GF user is randomly selected from all the GF users. {\color{black}Moreover, as HSIC can achieve better performance than FSIC strategy, the random selection SGF scheme with HSIC decoding strategy (termed as RS-SGF) is also applied as a benchmark.} 
	
	The final simulation results are obtained by averaging over $10^6$ independent trials. {\color{black} In each trial, the GB and GF users are randomly distributed in $\mathcal{D}_\text{B}$ and $\mathcal{D}_\text{F}$, respectively. Hereinafter, unless other specified, {\color{black}the simulation parameters are set similar with \cite{2018_Zhou_Dynamic decode and forward} as  $\alpha=3.8$, $K=4$, $D_\text{F}=D_1=800$ m, $D_0=300$ m, $R_{\text{B}}=1$ bps/Hz, and $R_{\text{F}}=0.9$ bps/Hz. The noise power is set as $-100$ dBm}.} {\color{black}Note that, for ease of calculating the theoretical results of BU-SGF and BU-SGF-PC schemes in (\ref{The1_3}) and (\ref{Theorem3}), respectively,  Gaussian-Chebyshev quadrature \cite{Gaussian_Chebyshev} is used to calculate the integrals.} And the complexity-accuracy trade-off parameters are set as $30$ \cite{2018_Zhou_Dynamic decode and forward}, which are sufficiently large to ensure the accuracy of the approximation expressions. We set $P_{\text{B}}=P_{\text{F}}=P_m$ for the schemes without power control (namely, BU-SGF, CS-SGF, and RS-SGF schemes).
	

	\begin{figure}[!t]
		\centering
		\subfigure[$R_\text{B}=1$ bps/Hz]{\includegraphics[width=0.49\textwidth]{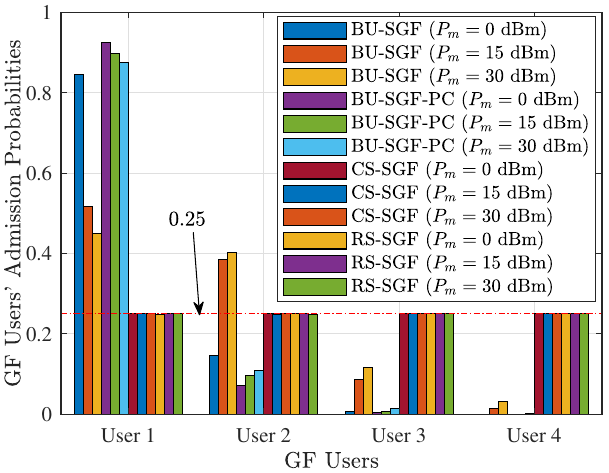}}
		\subfigure[$R_\text{B}=2$ bps/Hz]{\includegraphics[width=0.49\textwidth]{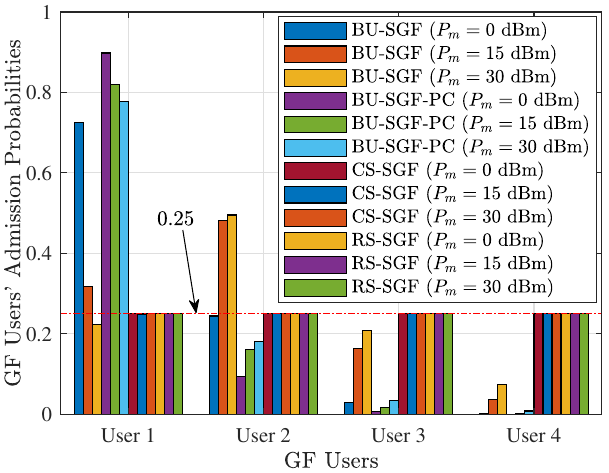}}
		\vspace{-10pt}
		\caption{\color{black}GF users' admission probabilities comparison of different schemes, where the distances from the 4 GF users to the BS vary from 200 m to 800 m with an interval of 200 m, and U$_{\text{B}}$'s distance to the BS is 200 m.}
		\label{fig1_admission_prob}
		\vspace{-20pt}
	\end{figure}

\subsection{Performance Comparison of Different Schemes}	
Fig. \ref{fig1_admission_prob} compares the GF users' admission probabilities of different schemes in 4-user case {\color{black}with different maximal transmit powers}.  We only plot the admission probability of the CS-SGF scheme, since CS-SGF-PC and CS-SGF schemes have the same fairness performance. As anticipated, both CS-SGF and RS-SGF schemes can achieve fair admission probability for each GF user {\color{black}(namely, every GF user can achieve an admission probability of 0.25)}. However, for the BU-SGF and BU-SGF-PC schemes, the GF users closer to the BS become more preferred to be admitted. That is because the BU-SGF (-PC) scheme prefers to admit the GF user with the highest data rate, and the user closer to the BS will be more likely to achieve {\color{black}a} higher rate.

	\begin{figure}[!t]
		\centering
		\subfigure[$D_1=500$ m]{\includegraphics[width=0.49\textwidth]{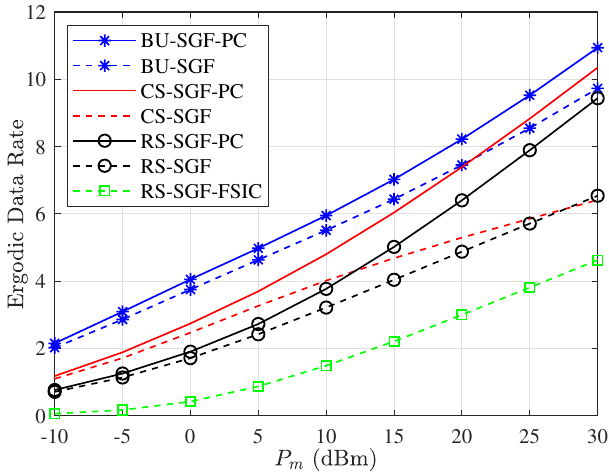}}
		\subfigure[$P_m=20$ dBm]{\includegraphics[width=0.49\textwidth]{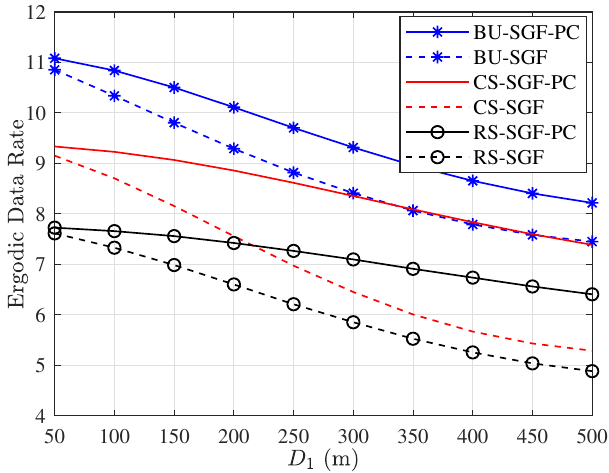}}	
		\vspace{-8pt}
		\caption{{GF users' ergodic data rate comparison of different SGF schemes, where $D_\text{F}=500$ m, $D_0=0$ m, and $R_\text{B}=1$ bps/Hz.}}
		\label{fig3_SR}
		\vspace{-15pt}
	\end{figure}

	Fig. \ref{fig3_SR}(a) depicts the ergodic data rates of different SGF schemes, where the GB and GF users are distributed in the same region (namely, $D_\text{F}=D_1=800$ m, $D_0=0$ m). {\color{black}As shown in \cite{Zhang_2020_WCL}, in order to avoid the GB user's decoding latency caused by SIC, the signal of the admitted GF user for RS-SGF-FSIC scheme \cite{Zhang_2020_WCL} is always decoded at the second stage of SIC.} It can be observed that, compared with RS-SGF-FSIC scheme, all the other schemes can achieve better ergodic data rate performance, where such improvement comes from effective use of multi-user diversity and/or the {\color{black}employment} of HSIC. It can also be observed that, with the proposed power control strategy, all the three SGF schemes' ergodic data rates can be effectively improved. Moreover, we can observe from Figs. \ref{fig1_admission_prob} and \ref{fig3_SR} that, compared with the BU-SGF (-PC) scheme, the CS-SGF (-PC) scheme can obtain fair admission probability at the price of slight data rate performance.

	
One interesting observation from Fig. \ref{fig3_SR}(a) is that the ergodic data rate of CS-SGF scheme is inferior to RS-SGF scheme in high SNR region. That is because, the impact of noise could be ignored in high SNR region, and the scheduled GF user in CS-SGF scheme will be more likely to be decoded at the first stage of SIC than that of RS-SGF scheme. The first decoded GF user will be severely interfered by the GB user's signal,  which leads a lower data rate. In order to verify this phenomenon, we carried out a relevant simulation in Fig. \ref{fig3_SR}(b).  {\color{black}It can be seen that, compared to RS-SGF scheme, the performance gain of CS-SGF scheme is reducing with the increasing of $D_1$. That is because when $D_1$ is much smaller than $D_\text{F}$, the admitted GF users in both schemes can rarely have larger channel gains than the GB user}, so that they will always be decoded at the second stage of SIC. In addition, with the increasing of $D_1$, the admitted GF user in CS-SGF scheme will be more likely to be decoded first than that of the RS-SGF scheme. However, by applying the proposed power control strategy, the CS-SGF-PC scheme can always achieve a much higher rate than RS-SGF (-PC) scheme by adjusting the decoding order adaptively, which demonstrates the importance of deploying the proposed power control strategy for CS-SGF scheme. Note that the BU-SGF scheme can always achieve better performance than CS-SGF and RS-SGF schemes, since it always admits the GF user with the largest data rate.

\begin{figure}[t]
	\centering
	\includegraphics[width=0.49\textwidth]{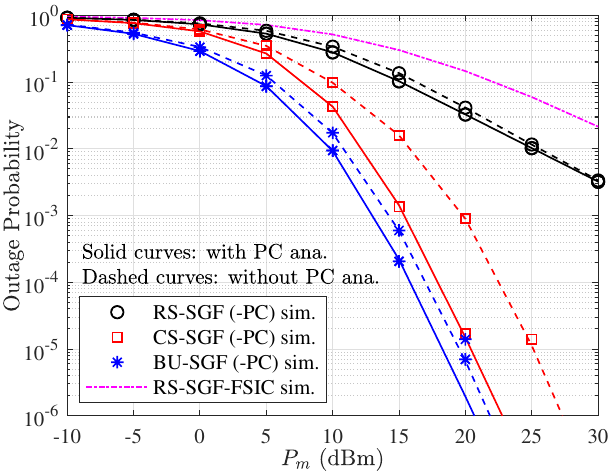}
			\vspace{-8pt}
	\caption{Outage probability comparison of different SGF schemes.}
	\label{fig2_dif_SGF}
			\vspace{-15pt}
\end{figure}
	
	The outage probabilities of different SGF schemes are compared in Fig. \ref{fig2_dif_SGF}. For ease of comparison with RS-SGF-FSIC scheme \cite{2020_Chao_Semi-GF} where the signal of the GF user is decoded at the first stage of SIC, $P_\text{B}$ is set as 10 dBm for RS-SGF-FSIC scheme \cite{2020_Chao_Semi-GF}. The RS-SGF-FSIC scheme shows comparative performance with RS-SGF-PC scheme with some impairment on the GB user's performance \cite{2020_Chao_Semi-GF}. {\color{black}Putting together the results in Figs. 2-4, we know that the CS-SGF (-PC) scheme underperforms the BU-SGF (-PC) scheme in terms of both the ergodic data rate and outage probability performances, but enables fairer access. Hence, the CS-SGF (-PC) scheme is more preferred for the scenarios with critical fairness requirements.}
	

	\begin{figure}[!t]
		\centering
		\subfigure[BU-SGF (-PC) scheme]{\includegraphics[width=0.49\textwidth]{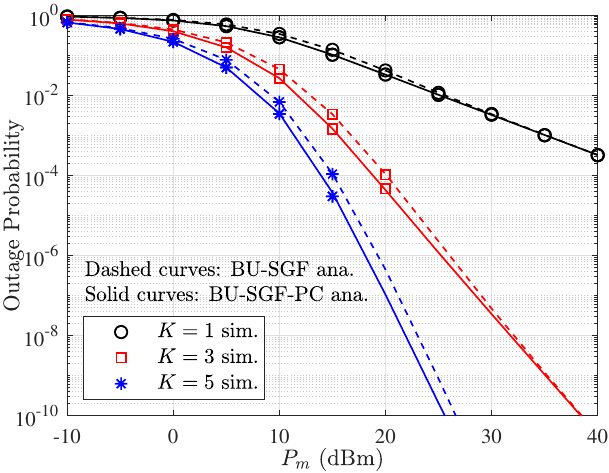}}
		\subfigure[CS-SGF (-PC) schemes]{\includegraphics[width=0.49\textwidth]{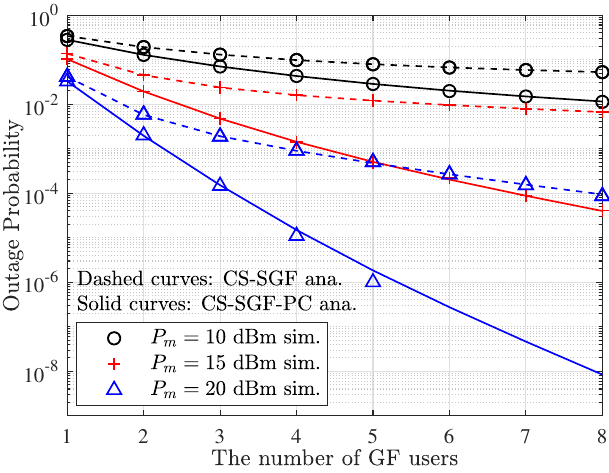}}	
		\vspace{-5pt}
		\caption{Outage probabilities of BU-SGF (-PC) and CS-SGF (-PC) schemes with different $K$.}
		\label{fig6_dif_K}
		\vspace{-15pt}
	\end{figure}
	

\subsection{Outage Performance Comparison with Different Parameters}	
	The outage probabilities of BU-SGF (-PC) scheme versus transmit SNR and CS-SGF (-PC) scheme versus the number of GF users are shown in Fig. \ref{fig6_dif_K}. It can be observed that the outage performances of the two schemes are improved with the increasing number of GF users, which verifies that the BU-SGF (-PC) scheme and the proposed CS-SGF (-PC) schemes can effectively utilize multi-user diversity. As anticipated in Remark 5, the proposed power control strategy can enhance the outage performance at the moderate SNR region, but the improvement vanishes in high SNR region.

	\begin{figure}[t]
	\centering
	\subfigure[BU-SGF-PC scheme]{\includegraphics[width=0.49\textwidth]{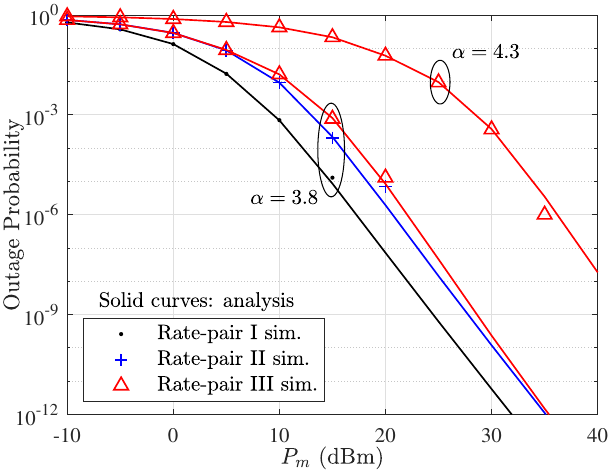}}
	\subfigure[CS-SGF-PC scheme]{\includegraphics[width=0.49\textwidth]{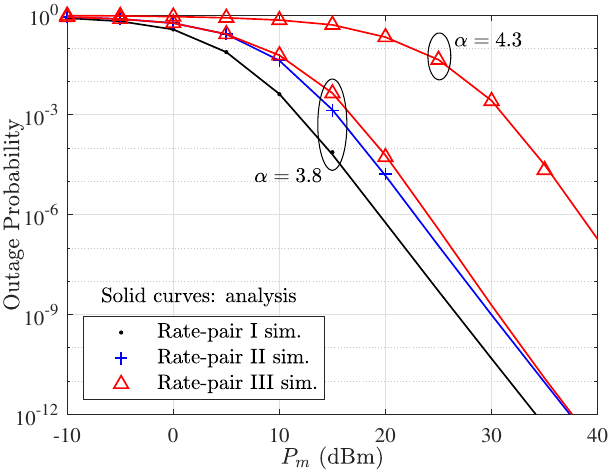}}
			\vspace{-8pt}
	\caption{Outage probability versus SNR with different rate pairs and path loss exponents. For rate-pair I, $R_\text{B}$ = 1 bps/Hz, $R_\text{F}$ = 0.5 bps/Hz; for rate-pair II, $R_\text{B}$ = 1 bps/Hz, $R_\text{F}$ = 0.9 bps/Hz; and for rate-pair III, $R_\text{B}$ = 1.5 bps/Hz, $R_\text{F}$ = 0.9 bps/Hz.}
	\label{fig7_rate_pair}
	\vspace{-15pt}
\end{figure}

	Fig. \ref{fig7_rate_pair} depicts the outage probability versus SNR with different path loss exponents and rate-pairs, where rate-pair denotes the pair of target data rates of the GB and GF users. Three rate-pairs are evaluated, specifically, for rate-pair I, $R_\text{B}=1$ bps/Hz, $R_\text{F}=0.5$ bps/Hz; for rate-pair II, $R_\text{B}=1$ bps/Hz, $R_\text{F}=0.9$ bps/Hz, and for rate-pair III, $R_\text{B}=1.5$ bps/Hz, $R_\text{F}=0.9$ bps/Hz. We can see that the outage probability increases with the increasing of GB/GF user's rate or the path loss exponent, {\color{black}which results from the demanding of GB/GF user's target rate or the increasing of path loss.} An interesting observation is that the outage probabilities of rate-pair II and rate-pair III superimpose at high SNR region for both BU-SGF-PC and CS-SGF-PC schemes. {\color{black}That can be explained by the high SNR approximation expressions in (\ref{The4_1}) and (\ref{lemma2_2}), in which  the terms with low exponent for both BU-SGF-PC and CS-SGF-PC schemes are proportional to GF user's target rate and irrelevant with the GB user's target rate.}

	\begin{figure}[t]
		\centering
		\subfigure[BU-SGF and BU-SGF-PC schemes]{\includegraphics[width=0.49\textwidth]{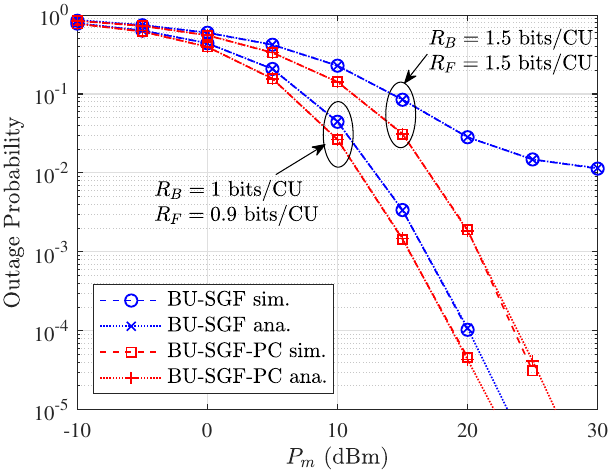}}
		\subfigure[CS-SGF and CS-SGF-PC schemes]{\includegraphics[width=0.49\textwidth]{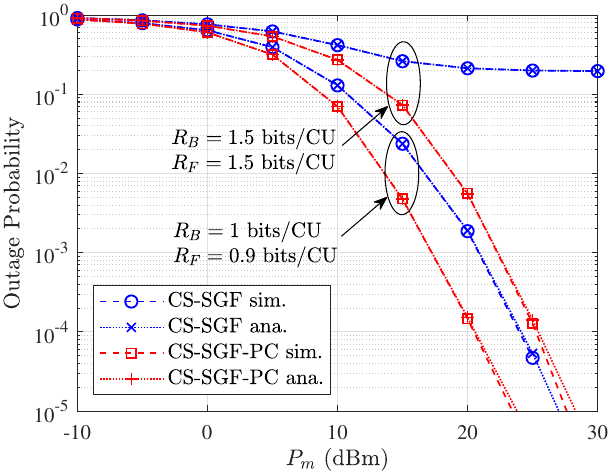}}
				\vspace{-8pt}
		\caption{Accuracy of analytical results, where $K=3$.}
		\label{fig4_ana_accuracy}
				\vspace{-15pt}
	\end{figure}

	\subsection{Verification the Accuracy of Theoretical Results}
	In Fig. \ref{fig4_ana_accuracy}, the accuracy of the analytical results is evaluated, which are based on Theorems \ref{theorem1} – \ref{theorem4}. The number of GF users are set as $K=3$. We can see from Fig. \ref{fig4_ana_accuracy} that the simulation matches well with the analytical results for both CS-SGF (-PC) and BU-SGF (-PC) schemes, which demonstrates the accuracy of analytical expressions in Theorems \ref{theorem1} – \ref{theorem4}.  We can also observe from Fig. \ref{fig4_ana_accuracy} that the proposed power control scheme can effectively improve the outage performance, especially for the case of $\gamma_{\text{B}}\gamma_{\text{F}}\geq1$.

	\begin{figure}[t]
		\centering
		\subfigure[BU-SGF and BU-SGF-PC schemes]{\includegraphics[width=0.49\textwidth]{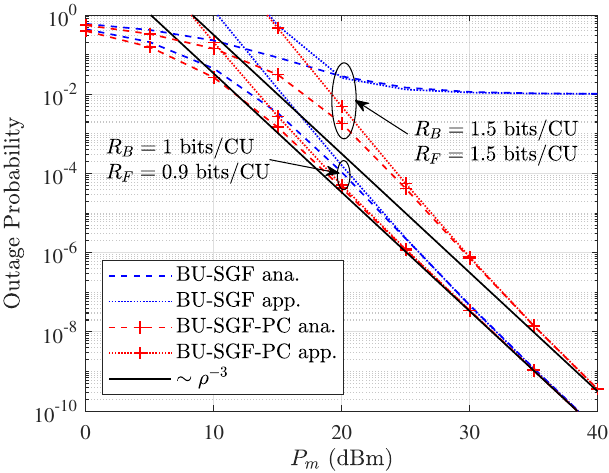}}
		\subfigure[CS-SGF and CS-SGF-PC schemes]{\includegraphics[width=0.49\textwidth]{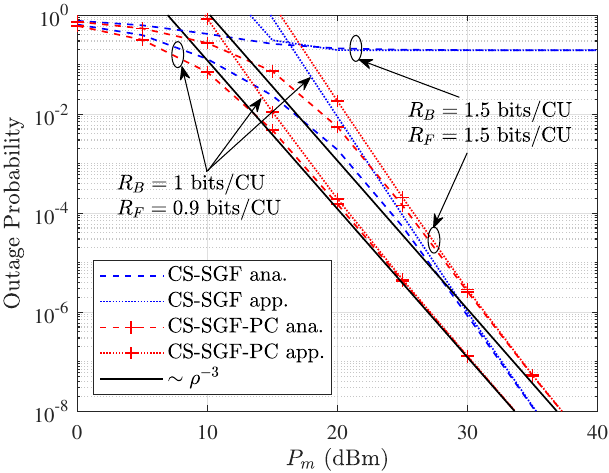}}
				\vspace{-8pt}
		\caption{Accuracy of high SNR approximation expressions, where $K=3$.}
		\label{fig5_approx_accuracy}
				\vspace{-15pt}
	\end{figure}

	Fig. \ref{fig5_approx_accuracy} evaluates the accuracy of {\color{black}the high SNR approximations} derived in (\ref{corollory 1}), (\ref{corollory3_1}), (\ref{The4_1}), and (\ref{lemma2_2}). {\color{black}Note that, the line $\rho^{-3}$ is plotted to facilitate the demonstration of the achievable diversity orders.} It can be observed from the two subfigures that all the high SNR approximations match well with the simulation results in high SNR region. {\color{black}As anticipated in Section \ref{Performance of two schemes}}, for both BU-SGF and CS-SGF schemes, zero diversity orders are achieved (namely, outage error floors exist) when the product of two users' target SINRs are larger than 1, or full diversity orders otherwise. However, as shown in the Fig. \ref{fig5_approx_accuracy} and {\color{black}demonstrated in Section \ref{Propose PC and performance analyze} that}, the BU-SGF-PC and CS-SGF-PC schemes can achieve full diversity orders even if the product of two users' target SINRs is larger than 1.


	\begin{figure}[t]
	\centering
	\subfigure[BU-SGF scheme]{\includegraphics[width=0.49\textwidth]{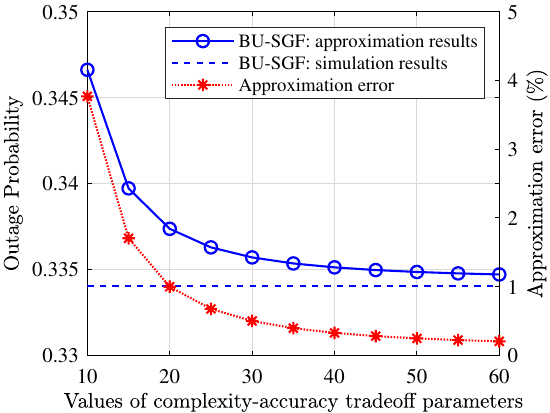}}
	\subfigure[CS-SGF scheme]{\includegraphics[width=0.49\textwidth]{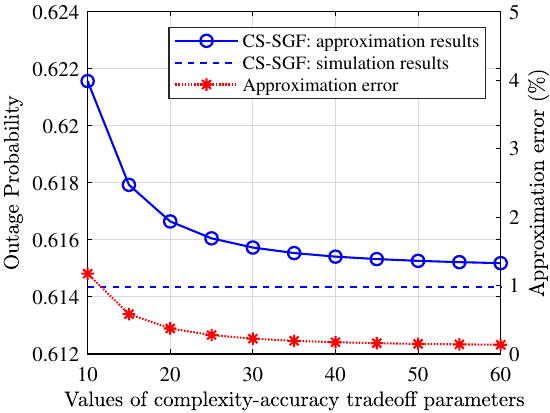}}
			\vspace{-8pt}
	\caption{The approximation error for applying Gaussian-Chebyshev quadrature, where  $P_m=0$ dBm.}
	\label{fig8_Gau_Che_accuraccy}
	\vspace{-15pt}
\end{figure}

{\color{black}In Fig. \ref{fig8_Gau_Che_accuraccy}, the approximation errors for applying Gaussian-Chebyshev quadrature are evaluated, where $P_m=0$ dBm, $R_\textnormal{B}$ = 1 bps/Hz, $R_\textnormal{F}$ = 0.9 bps/Hz. The approximation values for BU-SGF scheme are obtained by applying Gaussian-Chebyshev quadrature to (\ref{The1_3}a), and those for CS-SGF scheme are obtained by applying to (\ref{OP_theorem21}a).} The simulation results are obtained by averaging over $10^6$ independent trials. We can see that the approximation error is decreased with the increasing of the values of complexity-accuracy trade-off parameters, and the approximation value is very close to the simulation result (approximation error less than 1\%) when the complexity-accuracy trade-off parameters are set as 30.



	
	\section{Conclusions}\label{conclusion}
In this paper, we have studied the outage performances of two kinds of NOMA assisted SGF schemes in the system with randomly deployed users, namely, performance oriented SGF scheme and fairness oriented SGF scheme. A distributed power control strategy was proposed to effectively enhance the outage performance and relax the data rate constraint of the communication users for achieving full diversity orders. To facilitate performance evaluation, {\color{black}analytical expressions} of the admitted GF users' outage probabilities were developed for both CS-SGF (-PC) and BU-SGF (-PC) schemes. In addition, the achieved diversity orders of all considered schemes were derived, and {\color{black}analytical results showed that the schemes can always achieve full diversity orders by applying the proposed power control strategy}. Simulation results were provided to demonstrate the fairness superiority of the proposed CS-SGF (-PC) scheme, verify the effectiveness of the proposed power control strategy, and corroborate the accuracy of the analytical results.
	
Since full diversity orders can be obtained for these SGF schemes, in order to achieve better performance, more GF users should be allowed to contend for admission after considering some other constraints, such as access delay and coverage. In practice, the CS-SGF (-PC) scheme can be deployed if the operators care more about admission fairness. On the other hand, BU-SGF (-PC) scheme can be applied when the operators concern more about the rate performance. {\color{black}Since neither scheme is flexible, it is interesting to develop a new SGF scheme, which can dynamically strike a balance between the two extreme schemes by adjusting a weighting factor.} In addition, most existing work on SGF schemes only investigated single cell case, while the impact of inter-cell interference has not been well studied. In the future, it is promising to consider the multi-cell scenario, and analyze the corresponding performance by utilizing stochastic geometry.

\vspace{10pt}
	 \appendices
\section{Proof of Theorem \ref{theorem1}}\label{Proof of The1}
\setcounter{equation}{0}
\renewcommand{\theequation}{A.\arabic{equation}}

\subsection{Calculation of $T_0$}
Recall that $T_0$ can be expressed as
\begin{equation}\label{T10}
	\begin{aligned}
	T_0=
	&\mathbb{P}\left\{|g|^2>\alpha_\text{B},\rho_{\text{F}}|h_1|^2>\tau(|g|^2),\text{log}\left(1+\frac{\rho_{\text{F}}|h_K|^2}{\rho_{\text{B}}|g|^2+1}\right)<R_{\text{F}}\right\}\\
	=&\mathbb{P}\left\{|g|^2>\alpha_\text{B},|h_1|^2>\frac{\tau(|g|^2)}{\rho_{\text{F}}},|h_K|^2<\alpha_\text{F}(\rho_{\text{B}}|g|^2+1)\right\}.
	\end{aligned}
\end{equation}
Note that $|h_1|^2$ should be no larger than $|h_K|^2$, {\color{black}hence} there is a hidden constraint in (\ref{T10}) that 
$\frac{\tau(|g|^2)}{\rho_{\text{F}}}$ should be {\color{black}no larger} than $\alpha_\text{F}(\rho_{\text{B}}|g|^2+1)$. By applying $\tau(|g|^2)=\alpha_\text{B}^{-1}|g|^2-1$ and after some algebraic calculations, we {\color{black}know} $\frac{\tau(|g|^2)}{\rho_{\text{F}}}<\alpha_\text{F}(\rho_{\text{B}}|g|^2+1)$ holds under the following two cases
\begin{equation}\label{hidden constraint}
	\begin{cases}
|g|^2<\alpha_2, &\text{if} \ \gamma_{\text{B}}\gamma_{\text{F}}<1\\
\text{No constraint}, &\text{if} \ \gamma_{\text{B}}\gamma_{\text{F}}\geq 1
	\end{cases},
\end{equation}
where $\alpha_2=\frac{\alpha_\text{B}(1+\gamma_{\text{F}})}{1-\gamma_{\text{B}}\gamma_{\text{F}}}$. In the following, we will first focus the derivation on the case $\gamma_{\text{B}}\gamma_{\text{F}}<1$, and the case of $\gamma_{\text{B}}\gamma_{\text{F}}\geq 1$ will be derived at the end of this subsection. 

When $\gamma_{\text{B}}\gamma_{\text{F}}<1$, (\ref{T10}) can be rewritten as
\begin{equation}\label{T101}
\begin{aligned}
T_0=
\mathbb{P}&\left\{\alpha_\text{B}<|g|^2<\alpha_2,|h_1|^2>\frac{\tau(|g|^2)}{\rho_{\text{F}}},|h_K|^2<\alpha_\text{F}(\rho_{\text{B}}|g|^2+1)\right\},
\end{aligned}
\end{equation}
since $\alpha_\text{B}<\alpha_2$. Thus, $T_0$ is related to three random variables $|g|^2$, $|h_1|^2$, and $|h_K|^2$, where $|h_1|^2$ and $|h_K|^2$ are two order statistics, and $|g|^2$ is independent of them. The joint pdf of $|h_1|^2$ and $|h_K|^2$ can be expressed as \cite{2003_order_statistics}
\begin{equation}\label{jpdf_1m}
	\begin{aligned}
	f_{|h_1|^2,|h_K|^2}(x,y)=\tilde{\eta}_0f_\text{F}(x)f_\text{F}(y)\left[F_\text{F}(y)-F_\text{F}(x)\right]^{K-2},
	\end{aligned}
\end{equation}
where $x\leq y$ and $\tilde{\eta}_0=K(K-1)$. By applying (\ref{jpdf_1m}), $T_0$ can be expressed as
\begin{equation}
	\begin{aligned}
	T_0=&\int_{\alpha_\text{B}}^{\alpha_2}f_\text{B}(w)\int_{\frac{w}{\rho_{\text{F}}\alpha_\text{B}}-\frac{1}{\rho_\text{F}}}^{\alpha_\text{F}(\rho_{\text{B}}w+1)}\int_{x}^{\alpha_\text{F}(\rho_{\text{B}}w+1)}\tilde{\eta}_0f_\text{F}(x) f_\text{F}(y)\left[F_\text{F}(y)-F_\text{F}(x)\right]^{K-2}dydxdw.
	\end{aligned}
\end{equation}
We first calculate the following integral:
\begin{equation}\label{integration2}
\begin{aligned}
&\int_{\frac{w}{\rho_{\text{F}}\alpha_\text{B}}-\frac{1}{\rho_\text{F}}}^{\alpha_\text{F}(\rho_{\text{B}}w+1)}\int_{x}^{\alpha_\text{F}(\rho_{\text{B}}w+1)}f_\text{F}(x)f_\text{F}(y)\left[F_\text{F}(y)-F_\text{F}(x)\right]^{K-2}dydx\\
=&\int_{\frac{w}{\rho_{\text{F}}\alpha_\text{B}}-\frac{1}{\rho_\text{F}}}^{\alpha_\text{F}(\rho_{\text{B}}w+1)}\frac{1}{K-1}f_\text{F}(x)\left[F_\text{F}(\alpha_\text{F}\rho_{\text{B}}w+\alpha_\text{F})-F_\text{F}(x)\right]^{K-1}dx\\
=&\frac{1}{K(K-1)}\left[F_\text{F}(\alpha_\text{F}\rho_{\text{B}}w+\alpha_\text{F})-F_\text{F}\left(\frac{w}{\rho_{\text{F}}\alpha_\text{B}}-\frac{1}{\rho_\text{F}}\right)\right]^{K}.\\
\end{aligned}
\end{equation}
By using the result in (\ref{integration2}) and $\tau(|g|^2)=\alpha_\text{B}^{-1}|g|^2-1$, $T_0$ can be finally represented as
\begin{equation}\label{T03}
\begin{aligned}
T_0=&I_{1;0}+I_{2;0}.
\end{aligned}
\end{equation}

\subsection{Calculation of $T_k \ (1\leq k\leq K-1)$}
Recall that $T_k$ can be expressed as
\begin{equation}\label{Tm}
\begin{aligned}
T_k=
&\mathbb{P}\left\{|g|^2>\alpha_\text{B},E_k,\text{log}\left(1+\rho_{\text{F}}|h_k|^2\right)<R_{\text{F}},\text{log}\left(1+\frac{\rho_{\text{F}}|h_K|^2}{\rho_{\text{B}}|g|^2+1}\right)<R_{\text{F}}\right\}\\
=&\mathbb{P}\left\{|g|^2>\alpha_\text{B},\right.|h_k|^2<\alpha_\text{F},\rho_{\text{F}}|h_k|^2<\tau(|g|^2),\\
&\quad\left.\rho_{\text{F}}|h_{k+1}|^2>\tau(|g|^2),|h_K|^2<\alpha_\text{F}(\rho_{\text{B}}|g|^2+1)\right\}.
\end{aligned}
\end{equation}
Since $|h_{k+1}|^2$ should be no larger than $|h_K|^2$, by applying the hidden constraint in (\ref{hidden constraint}), $T_k$ can be rewritten as
\begin{equation}\label{Tm1}
\begin{aligned}
T_k=
&\mathbb{P}\left\{\alpha_\text{B}<|g|^2<\alpha_2,|h_k|^2<\alpha_\text{F},|h_k|^2<\rho_{\text{F}}^{-1}\tau(|g|^2),\right.\\
&\left.\qquad|h_{k+1}|^2>\rho_{\text{F}}^{-1}\tau(|g|^2),|h_K|^2<\alpha_\text{F}(\rho_{\text{B}}|g|^2+1)\right\},
\end{aligned}
\end{equation}
when $\gamma_{\text{B}}\gamma_{\text{F}}<1$.
 
In (\ref{Tm1}), one of the constraints of $|h_k|^2$ can be eliminated based on the value range of $|g|^2$. Note that
\begin{equation}\label{value range}
	\begin{cases}
	\alpha_\text{F}<\rho_{\text{F}}^{-1}\tau(|g|^2), &\text{if}\ |g|^2>\alpha_1\\
	\alpha_\text{F}\geq \rho_{\text{F}}^{-1}\tau(|g|^2), &\text{otherwise}
	\end{cases},
\end{equation}
where $\alpha_1=(1+\gamma_{\text{F}})\alpha_\text{B}$, and $\alpha_\text{B}<\alpha_1<\alpha_2=\frac{\alpha_\text{B}(1+\gamma_{\text{F}})}{1-\gamma_{\text{B}}\gamma_{\text{F}}}$. By applying (\ref{value range}), {\color{black}we have}
\begin{equation}\label{Tm2}
\begin{aligned}
T_k
=&\mathbb{P}\left\{\alpha_\text{B}<|g|^2<\alpha_1,|h_k|^2<\rho_{\text{F}}^{-1}\tau(|g|^2),|h_{k+1}|^2>\rho_{\text{F}}^{-1}\tau(|g|^2),|h_K|^2<\alpha_\text{F}(\rho_{\text{B}}|g|^2+1)\right\}\\
&+\mathbb{P}\left\{\alpha_1<|g|^2<\alpha_2,|h_k|^2<\alpha_\text{F},|h_{k+1}|^2>\rho_{\text{F}}^{-1}\tau(|g|^2),|h_K|^2<\alpha_\text{F}(\rho_{\text{B}}|g|^2+1)\right\}.\\
\end{aligned}
\end{equation}

We consider the following two cases to calculate $T_k$.
\subsubsection{$1\leq k\leq K-2$}
In this case, $|h_{k+1}|^2$ and $|h_K|^2$ are two different variables. The joint pdf of $|h_k|^2$, $|h_{k+1}|^2$, and $|h_K|^2$ is\cite{2003_order_statistics}
\begin{equation}\label{jpdf_3variables}
	\begin{aligned}
	&f_{|h_k|^2,|h_{k+1}|^2,|h_K|^2}(x,y,z)=\eta_kf_\text{F}(x)[F_\text{F}(x)]^{k-1} f_\text{F}(y)[F_\text{F}(z)-F_\text{F}(y)]^{K-k-2}f_\text{F}(z),
	\end{aligned}
\end{equation}
where $x\leq y \leq z$, and $\eta_k=\frac{K!}{(k-1)!(K-k-2)!}$.

Since $\text{U}_\text{B}$'s channel gain $|g|^2$ is independent with $|h_k|^2$, $|h_{k+1}|^2$, and $|h_K|^2$, by applying (\ref{jpdf_3variables}) and $\tau(|g|^2)=\alpha_\text{B}^{-1}|g|^2-1$, $T_k$ can be expressed as
\begin{equation}\label{Tm3}
\begin{aligned}
T_k
=&\int_{\alpha_\text{B}}^{\alpha_1}f_\text{B}(w)\int_{0}^{\frac{w}{\rho_{\text{F}}\alpha_\text{B}}-\frac{1}{\rho_\text{F}}}\eta_kf_\text{F}(x)[F_\text{F}(x)]^{k-1}\int_{\frac{w}{\rho_{\text{F}}\alpha_\text{B}}-\frac{1}{\rho_\text{F}}}^{\alpha_\text{F}(\rho_{\text{B}}w+1)}\int_{y}^{\alpha_\text{F}(\rho_{\text{B}}w+1)}f_\text{F}(y)\\
&\times[F_\text{F}(z)-F_\text{F}(y)]^{K-k-2}f_\text{F}(z)dzdydxdw\\
&+\int_{\alpha_1}^{\alpha_2}f_\text{B}(w)\int_{0}^{\alpha_\text{F}}\eta_kf_\text{F}(x)[F_\text{F}(x)]^{k-1}\int_{\frac{w}{\rho_{\text{F}}\alpha_\text{B}}-\frac{1}{\rho_\text{F}}}^{\alpha_\text{F}(\rho_{\text{B}}w+1)}\int_{y}^{\alpha_\text{F}(\rho_{\text{B}}w+1)}f_\text{F}(y)\\
&\times[F_\text{F}(z)-F_\text{F}(y)]^{K-k-2}f_\text{F}(z)dzdydxdw.\\
\end{aligned}
\end{equation}
By applying (\ref{integration2}) and after some manipulations, $T_k\ (1\leq k\leq K-2)$  can be finally expressed as
\begin{equation}\label{Tm4}
\begin{aligned}
T_k
=&\bar{\eta}_kI_{1;k}+\bar{\eta}_kI_{2;k}.\\
\end{aligned}
\end{equation}

\subsubsection{$k=K-1$}
In this case, we have $|h_{k+1}|^2=|h_K|^2$. Then
\begin{equation}\label{Tm5}
\begin{aligned}
T_{K-1}
=&\mathbb{P}\left\{\alpha_\text{B}<|g|^2<\alpha_1,|h_{K-1}|^2<\rho_{\text{F}}^{-1}\tau(|g|^2),\rho_{\text{F}}^{-1}\tau(|g|^2)<|h_K|^2<\alpha_\text{F}(\rho_{\text{B}}|g|^2+1)\right\}\\
&+\mathbb{P}\left\{\alpha_1<|g|^2<\alpha_2,|h_{K-1}|^2<\alpha_\text{F},\rho_{\text{F}}^{-1}\tau(|g|^2)<|h_K|^2<\alpha_\text{F}(\rho_{\text{B}}|g|^2+1)\right\}.\\
\end{aligned}
\end{equation}

The joint pdf of $|h_{K-1}|^2$ and $|h_K|^2$ is \cite{2003_order_statistics}
\begin{equation}\label{jpdf_2variables}
\begin{aligned}
&f_{|h_{K-1}|^2,|h_K|^2}(x,y)=\tilde{\eta}_0f_\text{F}(x)[F_\text{F}(x)]^{K-2}f_\text{F}(y),\\
\end{aligned}
\end{equation}
where $x\leq y$, and $\tilde{\eta}_0=K(K-1)$. By utilizing (\ref{jpdf_2variables}), $T_{K-1}$ can be further expressed as
\begin{equation}\label{Tm6}
	\begin{aligned}
	T_{K-1}
	=&\int_{\alpha_\text{B}}^{\alpha_1}f_\text{B}(w)\int_{0}^{\frac{w}{\rho_{\text{F}}\alpha_\text{B}}-\frac{1}{\rho_\text{F}}}\tilde{\eta}_0f_\text{F}(x)[F_\text{F}(x)]^{K-2}\int_{\frac{w}{\rho_{\text{F}}\alpha_\text{B}}-\frac{1}{\rho_\text{F}}}^{\alpha_\text{F}(\rho_{\text{B}}w+1)}f_\text{F}(y)dydxdw\\
	&+\int_{\alpha_1}^{\alpha_2}f_\text{B}(w)\int_{0}^{\alpha_\text{F}}\tilde{\eta}_0f_\text{F}(x)[F_\text{F}(x)]^{K-2}\int_{\frac{w}{\rho_{\text{F}}\alpha_\text{B}}-\frac{1}{\rho_\text{F}}}^{\alpha_\text{F}(\rho_{\text{B}}w+1)}f_\text{F}(y)dydxdw.\\
	\end{aligned}
\end{equation}
After some manipulations, $T_{K-1}$ can be finally represented as
\begin{equation}\label{Tm7}
\begin{aligned}
T_{K-1}
=&\bar{\eta}_{K-1}I_{1;K-1}+\bar{\eta}_{K-1}I_{2;K-1}.\\
\end{aligned}
\end{equation}

\subsection{Calculation of $T_K$}
We first rewrite $T_K$ as
\begin{equation}\label{TK}
	\begin{aligned}
	T_K=&\mathbb{P}\left\{|g|^2>\alpha_\text{B},E_K,\text{log}\left(1+\rho_{\text{F}}|h_K|^2\right)<R_{\text{F}}\right\}\\
	=&\mathbb{P}\left\{|g|^2>\alpha_\text{B},|h_K|^2<\rho_{\text{F}}^{-1}\tau(|g|^2),|h_K|^2<\alpha_\text{F}\right\}.\\
	\end{aligned}
\end{equation}
By utilizing (\ref{value range}), $T_K$ can be converted to
\begin{equation}
\begin{aligned}
T_K
=&\mathbb{P}\left\{\alpha_\text{B}<|g|^2<\alpha_1,|h_K|^2<\rho_{\text{F}}^{-1}\tau(|g|^2)\right\}+\mathbb{P}\left\{|g|^2>\alpha_1,|h_K|^2<\alpha_\text{F}\right\}.\\
\end{aligned}
\end{equation}
Since $|g|^2$ and $|h_K|^2$ are independent, we have
\begin{equation}\label{TK1}
	\begin{aligned}
	T_K=
	&\int_{\alpha_\text{B}}^{\alpha_1}f_\text{B}(w)\left[F_\text{F}\left(\frac{w}{\rho_{\text{F}}\alpha_\text{B}}-\frac{1}{\rho_\text{F}}\right)\right]^Kdw+\left(1-F_\text{B}(\alpha_1)\right)\left[F_\text{F}\left(\alpha_\text{F}\right)\right]^K\\
	=&I_{1,K}+I_{2,K}+I_4,
	\end{aligned}
\end{equation}
where the CDF of the largest order statistic $F_{|h_K|^2}(x)=\left[F_\text{F}(x)\right]^K$ \cite{2003_order_statistics} is applied.

\subsection{Calculation of $T_{K+1}$}
$T_{K+1}$ can be calculated following the similar way in deriving $T_K$ as
\begin{equation}\label{TK+1}
	\begin{aligned}
	T_{K+1}=
	&\mathbb{P}\left\{|g|^2<\alpha_\text{B},\text{log}\left(1+\frac{\rho_{\text{F}}|h_K|^2}{\rho_{\text{B}}|g|^2+1}\right)<R_{\text{F}}\right\}\\
	=&\mathbb{P}\left\{|g|^2<\alpha_\text{B},|h_K|^2<\alpha_\text{F}(\rho_{\text{B}}|g|^2+1)\right\}\\
	=&I_3.\\
	\end{aligned}
\end{equation}


By combining (\ref{T03}), (\ref{Tm4}), (\ref{Tm7}), (\ref{TK1}), and (\ref{TK+1}), we can get (\ref{The1_3}a).

{\color{black}We know from (\ref{hidden constraint}) that the expression of $\mathcal{P}_{\text{BU}}$ for the case of $\gamma_{\text{B}}\gamma_{\text{F}}\geq 1$ can be obtained by replacing $\alpha_2$ in (\ref{The1_3}a) with $\infty$. Therefore, when $\gamma_{\text{B}}\gamma_{\text{F}}\geq 1$,} we have (\ref{The1_3}b). The proof of Theorem \ref{theorem1} is completed.


\section{Proof of Corollary \ref{corollary1}}\label{Proof of The2}
The high SNR approximation of $\mathcal{P}_{\text{BU}}$ for the case $\gamma_{\text{B}}\gamma_{\text{F}}<1$ is derived first. It should be noted that, when $x\to 0$ and applying $e^{-x}\approx1-x$, the CDF of the GF users' unordered channel gains in (\ref{CDF_F_F}) can be approximated as 
\begin{equation}\label{CDF_F_approx}
	\begin{aligned}
	F_\text{F}(x)\approx\frac{1}{2}\sum_{l=1}^{L}\Psi_l\mu_lx=S_\text{F}x,
	\end{aligned}
\end{equation}
where $S_\text{F}=\frac{1}{2}\sum_{l=1}^{L}\Psi_l\mu_l$.
Similarly, when $y\to 0$, the pdf of the GB user's channel gain can be approximated as
\begin{equation}\label{pdf_B_approx}
\begin{aligned}
f_\text{B}(y)\approx\frac{1}{D_0+D_1}\sum_{n=1}^{N}\Phi_nc_n(1-c_ny).
\end{aligned}
\end{equation}

When $\rho_{\text{B}}=\rho_{\text{F}}\to \infty$, we have $\alpha_1=(1+\gamma_{\text{F}})\alpha_\text{B}\to 0$. And when $w\leq \alpha_1$, we have $\alpha_\text{F}\rho_{\text{B}}w+\alpha_\text{F}\leq\gamma_{\text{F}}\alpha_1+\alpha_\text{F}\to 0$ and $\frac{w}{\rho_{\text{F}}\alpha_\text{B}}-\frac{1}{\rho_\text{F}} \leq\alpha_\text{F}\to 0$. Therefore, by applying (\ref{CDF_F_approx}) and (\ref{pdf_B_approx}), the high SNR approximation of $I_{1;k}$ can be derived as
\begin{equation}\label{T0_approx}
\begin{aligned}
I_{1;k}\approx&\int_{\alpha_\text{B}}^{\alpha_1}\frac{S_{\text{F}}^k}{D_0+D_1}\sum_{n=1}^{N}\Phi_nc_n(1-c_nw)\left(\rho_{\text{F}}^{-1}\alpha_\text{B}^{-1}w-\rho_{\text{F}}^{-1}\right)^k\\
&\times\left[S_{\text{F}}(\alpha_\text{F}\rho_{\text{B}}w+\alpha_\text{F})-S_{\text{F}}\left(\rho_{\text{F}}^{-1}\alpha_\text{B}^{-1}w-\rho_{\text{F}}^{-1}\right)\right]^{K-k}dw\\
=&S_\text{B}S_{\text{F}}^K\int_{\alpha_\text{B}}^{\alpha_1}\sum_{i=0}^{K-k}\binom{K-k}{i}\left(\frac{\gamma_{\text{F}}+1}{\rho_{\text{F}}}\right)^{K-k-i}\\
&\times(\gamma_{\text{F}}w-\gamma_{\text{B}}^{-1}w)^i\sum_{j=0}^{k}\binom{k}{j}(\gamma_{\text{B}}^{-1}w)^{k-j}(-1)^j\rho_{\text{F}}^{-j}dw\\
=&\vec{I}_{1;k}.
\end{aligned}
\end{equation}

When $\rho_{\text{B}}=\rho_{\text{F}}\to \infty$, we have $\alpha_2=\frac{\alpha_\text{B}(\gamma_{\text{F}}+1)}{1-\gamma_{\text{B}}\gamma_{\text{F}}}\to 0$. If $w\leq \alpha_2$, we have $\alpha_\text{F}\rho_{\text{B}}w+\alpha_\text{F}\to 0$ and $\frac{w}{\rho_{\text{F}}\alpha_\text{B}}-\frac{1}{\rho_\text{F}} \to 0$. Similarly, by substituting (\ref{CDF_F_approx}) and (\ref{pdf_B_approx}) into $I_{2;k}$, we have
\begin{equation}\label{Tm_approx1}
\begin{aligned}
I_{2;k}
\approx
&\int_{\alpha_1}^{\alpha_2}\frac{S_{\text{F}}^K}{D_0+D_1}\sum_{n=1}^{N}\Phi_nc_n(1-c_nw)\alpha_\text{F}^k\left(\gamma_{\text{F}}w+\alpha_\text{F}-\frac{w}{\rho_{\text{F}}\alpha_\text{B}}-\frac{1}{\rho_\text{F}}\right)^{K-k}dw.\\
\end{aligned}
\end{equation}

After some algebraic manipulations, the high SNR approximation of $I_{2;k}$ becomes $\vec{I}_{2;k}$. Following the same steps as deriving the high SNR approximation of $I_{1;k}$, the high SNR approximations of $I_3$ and $I_4$ can be calculated as $\vec{I}_3$ and $\vec{I}_4$, respectively. Substituting the above results into (\ref{The1_3}a), we can obtain  (\ref{corollory 1}a).


According to the proof of Theorem \ref{theorem1}, for the case $\gamma_{\text{B}}\gamma_{\text{F}}\geq 1$, we only need to find the high SNR approximation of $I_{2;k}(\alpha_2\mapsto\infty)$. {\color{black}In the following, we first derive the} expression of $I_{2;k}(\alpha_2\mapsto\infty)$. By applying binomial theorem, $I_{2;k}(\alpha_2\mapsto\infty)$  {\color{black}can be denoted as}
\begin{equation}\label{I_2_2}
	\begin{aligned}
		I_{2;k}(\alpha_2\mapsto\infty)=
		&\left[F_\text{F}\left(\alpha_\text{F}\right)\right]^k\int_{\alpha_1}^{\infty}f_\text{B}(w)\sum_{m=0}^{K-k}\binom{K-k}{m}(-1)^m\\
		&\times F_\text{F}\left(\alpha_\text{F}\rho_{\text{B}}w+\alpha_\text{F}\right)^{K-k-m}F_\text{F}\left(\frac{w}{\rho_{\text{F}}\alpha_\text{B}}-\frac{1}{\rho_\text{F}}\right)^mdw.\\
	\end{aligned}
\end{equation} 

{\color{black}For ease of calculation}, we rewrite (\ref{CDF_F_F}) as
\begin{equation}\label{CDF_GF_re}
	\begin{aligned}
		F_\text{F}(x)\approx-\frac{1}{2}\sum_{l=0}^{L}\Psi_le^{-\mu_lx},
	\end{aligned}
\end{equation}
where $\Psi_0=-2$ and $\mu_0=0$. Based on (\ref{CDF_GF_re}), we can derive the following expression
\begin{equation}\label{F_F_M}
	\begin{aligned}
		\left[F_\text{F}(x)\right]^M
		&\approx\left(-\frac{1}{2}\sum_{l=0}^{L}\Psi_le^{-\mu_lx}\right)^M\\
		&\approx\left(-\frac{1}{2}\right)^M\sum_{\sum_{l=0}^{L}p_l=M}\binom{M}{p_0,\dots,p_L}\left(\prod_{l=0}^{L}\Psi_l^{p_l}\right)e^{-\sum_{l=0}^{L}p_l\mu_lx}.
	\end{aligned}
\end{equation}
Substituting (\ref{F_F_M}) into (\ref{I_2_2}), $I_{2;k}(\alpha_2\mapsto\infty)$ can be calculated as
\begin{equation}\label{I_2_3}
	\begin{aligned}
		I_{2;k}(\alpha_2\mapsto\infty)\approx&\frac{\left[F_\text{F}\left(\alpha_\text{F}\right)\right]^k}{D_0+D_1}\int_{\alpha_1}^{\infty}\sum_{n=1}^{N}\Phi_nc_ne^{-c_nw}\sum_{m=0}^{K-k}\binom{K-k}{m}(-1)^m\left(-\frac{1}{2}\right)^{K-k-m}\\
		&\times\sum_{\sum_{l=0}^{L}p_l={K-k-m}}\binom{{K-k-m}}{p_0,\dots,p_L}\left(\prod_{l=0}^{L}\Psi_l^{p_l}\right)e^{-\sum_{l=0}^{L}p_l\mu_l\left(\alpha_\text{F}\rho_{\text{B}}w+\alpha_\text{F}\right)}\\
		&\times\left(-\frac{1}{2}\right)^m\sum_{\sum_{l=0}^{L}q_l=m}\binom{m}{q_0,\dots,q_L}\left(\prod_{l=0}^{L}\Psi_l^{q_l}\right)e^{-\sum_{l=0}^{L}q_l\mu_l\left(\frac{w}{\rho_{\text{F}}\alpha_\text{B}}-\frac{1}{\rho_\text{F}}\right)}dw.\\
	\end{aligned}
\end{equation}
After some algebraic manipulations, we have $I_{2;k}(\alpha_2\mapsto\infty)\approx H_{2;k}$, where 
\begin{equation}\label{H_kzhi}	
\begin{aligned}
H_{2;k}=\Omega\left[F_\text{F}\left(\alpha_\text{F}\right)\right]^ke^{\sum_{l=0}^{L}\left(q_l\mu_l\rho_{\text{F}}^{-1}-p_l\mu_l\alpha_\text{F}\right)}\frac{e^{-\left[\sum_{l=0}^{L}\left(p_l\mu_l\alpha_\text{F}\rho_{\text{B}}+q_l\mu_l\rho_{\text{F}}^{-1}\alpha_\text{B}^{-1}\right)+c_n\right]\alpha_1}}{\sum_{l=0}^{L}\left(p_l\mu_l\alpha_\text{F}\rho_{\text{B}}+q_l\mu_l\rho_{\text{F}}^{-1}\alpha_\text{B}^{-1}\right)+c_n}.\\
\end{aligned}
\end{equation}
 
When $\rho_{\text{B}}=\rho_{\text{F}}\to \infty$, we have the following two approximations: $\sum_{l=0}^{L}\left(q_l\mu_l\rho_{\text{F}}^{-1}-p_l\mu_l\alpha_\text{F}\right)\to 0$ and $\left[\sum_{l=0}^{L}\left(p_l\mu_l\alpha_\text{F}\rho_{\text{B}}+q_l\mu_l\rho_{\text{F}}^{-1}\alpha_\text{B}^{-1}\right)+c_n\right]\alpha_1\to 0$. Substituting the above two approximations into (\ref{H_kzhi})  and applying $e^{-x}\approx1-x$, we can find that the high SNR approximation of $I_{2;k}(\alpha_2\mapsto\infty)$ is $\vec{H}_{2;k}$. Combing the above results, we can derive  (\ref{corollory 1}b), and the proof of Corollary \ref{corollary1} is completed.

\section{Proof of Theorem \ref{theorem2}}\label{Proof SC-SGF}
Firstly, we find that $\Delta_1$ can be rewritten as
\begin{equation}\label{Delta_1_1}
\begin{aligned}
\Delta_1=&\mathbb{P}\left\{|g|^2<\alpha_\text{B},|h|^2<\alpha_\text{F}\left(\rho_{\text{B}}|g|^2+1\right)\right\}.
\end{aligned}
\end{equation}

{\color{black}According to \cite{2020_Lu_TVT}, the CDF of the admitted GF user's channel gain in CS-SGF scheme can be expressed as
\begin{equation}\label{CDF_GF}
	\begin{aligned}
		F^\text{CS}_\text{F}(x)
		\approx&\frac{1}{2}\sum_{l=1}^{L}\Psi_l(1-e^{-\mu_lx})^K.\\
	\end{aligned}
\end{equation}}
Since the channels of the GB user $g$ and the admitted GF user $h$ are independent, by substituting (\ref{pdf_GB}) and (\ref{CDF_GF}) into (\ref{Delta_1_1}), $\Delta_1$ can be approximated as
\begin{equation}\label{Delta_1_2}
\begin{aligned}
\Delta_1\approx&\int_{0}^{\alpha_\text{B}}\Xi_2\left[1-e^{-\mu_l\alpha_\text{F}(\rho_{\text{B}}w+1)}\right]^Kc_ne^{-c_nw}dw.
\end{aligned}
\end{equation}
By applying binomial theorem, $\Delta_1$ can be further approximated as
\begin{equation}\label{Delta_1_3}
\begin{aligned}
\Delta_1\approx&\Xi_2c_n\sum_{k=0}^{K}\binom{K}{k}(-1)^k\int_{0}^{\alpha_\text{B}}e^{-k\mu_l\alpha_\text{F}(\rho_{\text{B}}w+1)-c_nw}dw\\
=&\frac{\Xi_1c_n}{\Theta_1}e^{-k\mu_l\alpha_\text{F}}\left(1-e^{-\Theta_1\alpha_\text{B}}\right).
\end{aligned}
\end{equation}

For $\Delta_2$, it can be expressed as 
\begin{equation}\label{Delta_2_1}
\begin{aligned}
\Delta_2=&\mathbb{P}\left\{|g|^2>\alpha_\text{B},\frac{\tau(|g|^2)}{\rho_\text{F}}<|h|^2<\alpha_\text{F}(\rho_{\text{B}}|g|^2+1)\right\}.
\end{aligned}
\end{equation}
Applying the hidden constraint in (\ref{hidden constraint}), we calculate $\Delta_2$ by considering two cases. For the first case of $\gamma_{\text{B}}\gamma_{\text{F}}<1$, $\Delta_2$ becomes
\begin{equation}\label{Delta_2_3}
\begin{aligned}
\Delta_2=&\mathbb{P}\left\{\alpha_\text{B}<|g|^2<\alpha_2,\frac{\tau(|g|^2)}{\rho_\text{F}}<|h|^2<\alpha_\text{F}(\rho_{\text{B}}|g|^2+1)\right\}.
\end{aligned}
\end{equation}
By applying (\ref{pdf_GB}) and (\ref{CDF_GF}), and substituting $\tau(|g|^2)=\alpha_\text{B}^{-1}|g|^2-1$, $\Delta_2$ can be approximated as
\begin{equation}\label{Delta_2_6}
\begin{aligned}
\Delta_2
\approx&\int_{\alpha_\text{B}}^{\alpha_2}\left[(1-e^{-\mu_l\alpha_\text{F}(\rho_{\text{B}}w+1)})^K-(1-e^{-\frac{\mu_l(w-\alpha_\text{B})}{\rho_\text{F}\alpha_\text{B}}})^K\right]\\
&\times \frac{1}{2}\sum_{l=1}^{L}\Psi_l\frac{1}{D_1+D_0}\sum_{n=1}^{N}\Phi_nc_ne^{-c_nw}dw.
\end{aligned}
\end{equation}
After some manipulations, the approximation of $\Delta_2$ can be simplified as
\begin{equation}\label{Delta_2_4}
\begin{aligned}
\Delta_2
\approx&\frac{\Xi_1c_n}{\Theta_1}e^{-k\mu_l\alpha_\text{F}}(e^{-\Theta_1\alpha_\text{B}}-e^{-\Theta_1\alpha_2})-\frac{\Xi_1c_n}{\Theta_2}e^{\frac{k\mu_l}{\rho_{\text{F}}}}(e^{-\Theta_2\alpha_\text{B}}-e^{-\Theta_2\alpha_2}).
\end{aligned}
\end{equation}

For the second case of $\gamma_{\text{B}}\gamma_{\text{F}}\geq 1$, $\Delta_2$ becomes
\begin{equation}\label{Delta_2_10}
\begin{aligned}
\Delta_2=&\mathbb{P}\left\{|g|^2>\alpha_\text{B},\frac{\tau(|g|^2)}{\rho_\text{F}}<|h|^2<\alpha_\text{F}(\rho_{\text{B}}|g|^2+1)\right\}.
\end{aligned}
\end{equation}
Similarly, by substituting (\ref{pdf_GB}) and (\ref{CDF_GF}), and after some manipulations, $\Delta_2$ can be approximated as
\begin{equation}\label{Delta_2_9}
\begin{aligned}
\Delta_2
\approx&\frac{\Xi_1c_n}{\Theta_1}e^{-k\mu_l\alpha_\text{F}}{e^{-\Theta_1\alpha_\text{B}}}-\frac{\Xi_1c_n}{\Theta_2}e^{\frac{k\mu_l}{\rho_\text{F}}}{e^{-\Theta_2\alpha_\text{B}}}.\\
\end{aligned}
\end{equation}

For $\Delta_3$, by applying (\ref{value range}), it can be rewritten as
\begin{equation}\label{Delta_3_1}
\begin{aligned}
\Delta_{3}=&\mathbb{P}\left\{\alpha_\text{B}<|g|^2<\alpha_1,|h|^2<\frac{\tau(|g|^2)}{\rho_{\text{F}}}\right\}+\underbrace{\mathbb{P}\left\{|g|^2>\alpha_1,|h|^2<\alpha_\text{F}\right\}}_{A_1}.
\end{aligned}
\end{equation}
By applying (\ref{CDF_unordered}) and (\ref{CDF_GF}), $A_1$ can be approximated as
\begin{equation}\label{Delta_3_2}
\begin{aligned}
A_1\approx&\frac{1}{2}\left[1-\frac{1}{D_1+D_0}\sum_{n=1}^{N}\Phi_n\left(1-e^{-c_n\alpha_1}\right)\right]\sum_{l=1}^{L}\Psi_l(1-e^{-\mu_l\alpha_\text{F}})^K\\
=&\Xi_2e^{-c_n\alpha_1}\left(1-e^{-\mu_l\alpha_\text{F}}\right)^K,
\end{aligned}
\end{equation}
where the equality holds by using $\frac{1}{D_1+D_0}\sum_{n=1}^{N}\Phi_n=1$. Following the previous derivation procedure and considering (\ref{Delta_3_2}), $\Delta_3$ can be approximated as
\begin{equation}\label{Delta_3_3}
\begin{aligned}
\Delta_{3}\approx&\frac{\Xi_1c_n}{\Theta_2}e^{\frac{k\mu_l}{\rho_{\text{F}}}}\left(e^{-\Theta_2\alpha_\text{B}}-e^{-\Theta_2\alpha_1}\right)+\Xi_2e^{-c_n\alpha_1}\left(1-e^{-\mu_l\alpha_\text{F}}\right)^K.
\end{aligned}
\end{equation}

Finally, we can obtain (\ref{OP_theorem21}) by combining (\ref{Delta_1_3}), (\ref{Delta_2_4}), and (\ref{Delta_3_3}), and combining (\ref{Delta_1_3}), (\ref{Delta_2_9}), and (\ref{Delta_3_3}). The proof is complete.

\section{Proof of Corollary \ref{corollary3}}\label{Proof of corollary3}
	The high SNR approximation of $\mathcal{P}_{\text{CS}}$ is derived based on (\ref{Delta_1_2}), (\ref{Delta_2_6}), (\ref{Delta_2_9}), and (\ref{Delta_3_1}). When $\rho_{\text{B}}=\rho_{\text{F}}\rightarrow\infty$, we have $\alpha_\text{B}\to 0$ and $\alpha_2\to 0$. By applying $e^{-x}\approx1-x$ for $x\rightarrow0$ to (\ref{Delta_1_2}), the high SNR approximation of $\Delta_1$ can be expressed as
	\begin{equation}\label{Delta_1_approx0}
		\begin{aligned}
			\Delta_1
			\approx&\Xi_2c_n\mu_l^K\alpha_\text{F}^K\int_{0}^{\alpha_\text{B}}(\rho_{\text{B}}x+1)^Kdx.
		\end{aligned}
	\end{equation}
	By applying binomial theorem, (\ref{Delta_1_approx0}) can be further derived as
	\begin{equation}\label{Delta_1_approx}
		\begin{aligned}
			\Delta_1
			\approx&\Xi_2c_n\mu_l^K\alpha_\text{F}^K\int_{0}^{\alpha_\text{B}}\sum_{k=0}^{K}\binom{K}{k}\rho_{\text{B}}^kx^kdx\\
			=&\frac{\Xi_2c_n}{\rho_{\text{B}}}\left(\frac{\mu_l\gamma_{\text{F}}}{\rho_{\text{B}}}\right)^K\sum_{k=0}^{K}\binom{K}{k}\frac{\gamma_{\text{B}}^{k+1}}{k+1}.
		\end{aligned}
	\end{equation}

	Using (\ref{Delta_2_6}), the high SNR approximation of $\Delta_2$ in the case of $\gamma_{\text{B}}\gamma_{\text{F}}<1$ can be derived as
	\begin{equation}\label{Delta_2_approx1}
		\begin{aligned}
			\Delta_2
			\approx&\Xi_2c_n\int_{\alpha_\text{B}}^{\alpha_2}\left(1-e^{-\frac{\mu_l\gamma_{\text{F}}(\rho_{\text{B}}y+1)}{\rho_\text{F}}}\right)^Kdy-\Xi_2c_n\int_{\alpha_\text{B}}^{\alpha_2}\left(1-e^{-\frac{\mu_l(y-\alpha_{\text{B}})}{\rho_{\text{F}}\alpha_{\text{B}}}}\right)^Kdy\\
			=&\frac{\Xi_2c_n\mu_l^K\gamma_{\text{F}}^K}{\rho_{\text{B}}^{K+1}}\sum_{k=0}^{K}\binom{K}{k}\frac{\tilde{\alpha}_2^{k+1}-\gamma_{\text{B}}^{k+1}}{k+1}-\frac{\Xi_2c_n\mu_l^K}{\rho_{\text{B}}^{K+1}}\sum_{k=0}^{K}\binom{K}{k}(-1)^{K-k}\frac{\tilde{\alpha}_2^{k+1}-\gamma_{\text{B}}^{k+1}}{\gamma_{\text{B}}^k(k+1)}.\\
		\end{aligned}
	\end{equation}
	
	By applying (\ref{Delta_2_9}), the high SNR approximation of $\Delta_2$ in the case of $\gamma_{\text{B}}\gamma_{\text{F}}\geq 1$ can be expressed as
	\begin{equation}\label{Gamma_2}
		\begin{aligned}
			\Delta_{2}
			\approx&\Xi_1c_n\left(\Theta_1^{'-1}-\Theta_2^{'-1}\right).
		\end{aligned}
	\end{equation}
	
	Using (\ref{Delta_3_1}) and following the same lines of deriving (\ref{Delta_1_approx}), we can obtain the approximation of $\Delta_3$ as
	\begin{equation}\label{Delta_3_approx}
		\begin{aligned}
			\Delta_3
			\approx&\frac{\Xi_2c_n\mu_l^K\gamma_{\text{B}}}{\rho_{\text{B}}^{K+1}}\sum_{k=0}^{K}\binom{K}{k}(-1)^{K-k}\frac{(1+\gamma_{\text{F}})^{k+1}-1}{k+1}+\Xi_2\left(\frac{\mu_l\gamma_{\text{F}}}{\rho_{\text{B}}}\right)^K.
		\end{aligned}
	\end{equation}
	
	Finally, we can obtain (\ref{corollory3_1}) by combining (\ref{Delta_1_approx}), (\ref{Delta_2_approx1}), and (\ref{Delta_3_approx}). Similarly, we can derive (\ref{coroll33_2}) by combining (\ref{Delta_1_approx}), (\ref{Gamma_2}), and (\ref{Delta_3_approx}). The proof is complete.

\section{Proof of Theorem \ref{theorem3}}\label{Proof of The3}
As all the GF users' maximal transmit SNRs are assumed to be $\rho_m$ and their channel gains are ordered as (\ref{channel order}), $T_0^{\text{PC}}$ can be rewritten as
\begin{equation}\label{T0_PC1}
\begin{aligned}
T_0^{\text{PC}}=
&\mathbb{P}\left\{|g|^2>\alpha'_\text{B},\rho_m|h_1|^2>\tau'(|g|^2),\frac{\rho_m|h_K|^2}{\rho_{m}|g|^2+1}>\tau'(|g|^2),\frac{\rho_m|h_K|^2}{\rho_{m}|g|^2+1}<\gamma_{\text{F}}\right\}\\
&+\mathbb{P}\left\{|g|^2>\alpha'_\text{B},\rho_m|h_1|^2>\tau'(|g|^2),\frac{\rho_m|h_K|^2}{\rho_{m}|g|^2+1}<\tau'(|g|^2),\tau'(|g|^2)<\gamma_{\text{F}}\right\}.
\end{aligned}
\end{equation}

After some manipulations, $T_0^{\text{PC}}$ can be rewritten as
\begin{equation}\label{T0_PC2}
\begin{aligned}
T_0^{\text{PC}}=
&\mathbb{P}\left\{\alpha'_\text{B}<|g|^2<\alpha'_1,|h_1|^2>\frac{\tau'(|g|^2)}{\rho_m},\frac{\rho_{m}|g|^2+1}{\rho_m}\tau'(|g|^2)<|h_K|^2<\frac{\rho_{m}|g|^2+1}{\rho_m}\gamma_{\text{F}}\right\}\\
&+\mathbb{P}\left\{\alpha'_\text{B}<|g|^2<\alpha'_1,|h_1|^2>\frac{\tau'(|g|^2)}{\rho_m},|h_K|^2<\frac{\rho_{m}|g|^2+1}{\rho_m}\tau'(|g|^2)\right\},\\
\end{aligned}
\end{equation}
where $\alpha'_1=\alpha'_\text{B}(1+\gamma_{\text{F}})$, and the constraint $|g|^2<\alpha'_1$ is obtained due to $\tau'(|g|^2)<\gamma_{\text{F}}$ and $\tau'(|g|^2)=(\alpha'_\text{B})^{-1}|g|^2-1$. Surprisingly, we find that the two terms in (\ref{T0_PC2}) can be combined, namely, $T_0^{\text{PC}}$ can be further simplified as
\begin{equation}\label{T0_PC3}
\begin{aligned}
T_0^{\text{PC}}=
&\mathbb{P}\left\{\alpha'_\text{B}<|g|^2<\alpha'_1,|h_1|^2>\frac{\tau'(|g|^2)}{\rho_m},|h_K|^2<\frac{\rho_{m}|g|^2+1}{\rho_m}\gamma_{\text{F}}\right\}.\\
\end{aligned}
\end{equation}

By comparing the calculation process from (\ref{T10}) to (\ref{T101}) with that from (\ref{T0_PC1}) to (\ref{T0_PC3}), we can find that the constraint of $\gamma_{\text{B}}\gamma_{\text{F}}<1$ required for the derivation of (\ref{T101}) is not required here any more. Following the same steps of deriving (\ref{T03}), we have $T_0^{\text{PC}}=I_{1;0}(\rho_\text{F}\mapsto \rho_m)$.

Then, $T_k^{\text{PC}}(1\leq k\leq K-1)$ can be rewritten as
\begin{equation}\label{Tk_PC1}
\begin{aligned}
T_k^{\text{PC}}=&
\mathbb{P}\left\{|g|^2>\alpha'_\text{B},\rho_m|h_k|^2<\gamma_{\text{F}},\rho_m|h_k|^2<\tau'(|g|^2),\right.\\
&\left.\rho_m|h_{k+1}|^2>\tau'(|g|^2),\tau'(|g|^2)<\frac{\rho_m|h_K|^2}{\rho_{m}|g|^2+1}<\gamma_{\text{F}}\right\}\\
&+\mathbb{P}\left\{|g|^2>\alpha'_\text{B},|h_k|^2<\frac{\tau'(|g|^2)}{\rho_m},|h_{k+1}|^2>\frac{\tau'(|g|^2)}{\rho_m},\right.\\
&\left. \rho_m|h_k|^2<\gamma_{\text{F}},\frac{\rho_m|h_K|^2}{\rho_{m}|g|^2+1}<\tau'(|g|^2),\tau'(|g|^2)<\gamma_{\text{F}}\right\}.
\end{aligned}
\end{equation}

After some algebraic operations, (\ref{Tk_PC1}) can be converted to
\begin{equation}\label{Tk_PC2}
\begin{aligned}
T_k^{\text{PC}}=
&\mathbb{P}\left\{\alpha'_\text{B}<|g|^2<\alpha'_1,|h_k|^2<\frac{\tau'(|g|^2)}{\rho_m},|h_{k+1}|^2>\frac{\tau'(|g|^2)}{\rho_m},\right.\\
&\left.|h_k|^2<\frac{\gamma_{\text{F}}}{\rho_m},\frac{\rho_{m}|g|^2+1}{\rho_m}\tau'(|g|^2)<|h_K|^2<\frac{\rho_{m}|g|^2+1}{\rho_m}\gamma_{\text{F}}\right\}\\
&+\mathbb{P}\left\{\alpha'_\text{B}<|g|^2<\alpha'_1,|h_k|^2<\frac{\tau'(|g|^2)}{\rho_m},|h_k|^2<\frac{\gamma_{\text{F}}}{\rho_m},\right.\\
&\left. |h_{k+1}|^2>\frac{\tau'(|g|^2)}{\rho_m},|h_K|^2<\frac{\rho_{m}|g|^2+1}{\rho_m}\tau'(|g|^2)\right\}.
\end{aligned}
\end{equation}

Combining the two parts of (\ref{Tk_PC2}) together, we have
\begin{equation}\label{Tk_PC3}
\begin{aligned}
T_k^{\text{PC}}=
&\mathbb{P}\left\{\alpha'_\text{B}<|g|^2<\alpha'_1,|h_k|^2<\frac{\tau'(|g|^2)}{\rho_m},|h_k|^2<\frac{\gamma_{\text{F}}}{\rho_m},\right.\\
&\left. |h_{k+1}|^2>\frac{\tau'(|g|^2)}{\rho_m},|h_K|^2<\frac{\rho_{m}|g|^2+1}{\rho_m}\gamma_{\text{F}}\right\}.
\end{aligned}
\end{equation}
Since $|g|^2<\alpha'_1$, by applying (\ref{value range}), we can further simplify $T_k^{\text{PC}}$ as
\begin{equation}\label{Tk_PC4}
\begin{aligned}
T_k^{\text{PC}}=
&\mathbb{P}\left\{\alpha'_\text{B}<|g|^2<\alpha'_1,|h_k|^2<\frac{\tau'(|g|^2)}{\rho_m}, |h_{k+1}|^2>\frac{\tau'(|g|^2)}{\rho_m},|h_K|^2<\frac{\rho_{m}|g|^2+1}{\rho_m}\gamma_{\text{F}}\right\}.
\end{aligned}
\end{equation}

Following the same lines of deriving (\ref{Tm4}) and (\ref{Tm7}), we can obtain $T_k^{\text{PC}}\ =I_{1;k}(\rho_\text{F}\mapsto \rho_m,\rho_\text{B}\mapsto \rho_m)$ for $(1\leq k\leq K-1)$. Similar with the derivation of (\ref{TK1}) and (\ref{TK+1}), we have $T_K^{\text{PC}}\ =I_{1;K}(\rho_\text{F}\mapsto \rho_m,\rho_\text{B}\mapsto \rho_m)+\left[1-F_\text{B}(\alpha'_1)\right]\left[F_\text{F}\left(\alpha'_{\text{F}}\right)\right]^K$ and $T_{K+1}^{\text{PC}}\ =F_\text{B}(\alpha'_\text{B})[F_\text{F}(\alpha'_\text{F})]^K$. Combining all the above results, we can get (\ref{Theorem3}), and the proof is complete.

%

\section{Proof of Theorem \ref{theorem4}}\label{Proof of The4}
We simplify (\ref{OP_PC_1}) first. Note that $\Delta_4$ and $\Delta_5$ can be respectively converted to
\begin{equation}\label{Delta_4_1}
\begin{aligned}
\Delta_4=&
\mathbb{P}\left\{\alpha'_{\text{B}}<|g|^2<\alpha'_{1},\frac{\rho_{m}|g|^2+1}{\rho_m}\tau'(|g|^2)<|h|^2<\frac{\rho_{m}|g|^2+1}{\rho_m}\gamma_{\text{F}}\right\}
\end{aligned}
\end{equation}
and
\begin{equation}\label{Delta_5_1}
\begin{aligned}
\Delta_5=
&\mathbb{P}\left\{\alpha'_{\text{B}}<|g|^2<\alpha'_{1},\frac{\tau'(|g|^2)}{\rho_m}<|h|^2<\frac{\rho_{m}|g|^2+1}{\rho_m}\tau'(|g|^2)\right\}.
\end{aligned}
\end{equation}
Combining (\ref{Delta_4_1}) and (\ref{Delta_5_1}), $\Delta_{4}+\Delta_{5}$ can be represented as
\begin{equation}\label{Delta_45}
\begin{aligned}
\Delta_{4}+\Delta_{5}=&
\mathbb{P}\left\{\alpha'_{\text{B}}<|g|^2<\alpha'_{1},\frac{\tau'(|g|^2)}{\rho_m}<|h|^2<\frac{\rho_{m}|g|^2+1}{\rho_m}\gamma_{\text{F}}\right\}.
\end{aligned}
\end{equation}

By applying (\ref{value range}), $\Delta_{3}$ can be rewritten as
\begin{equation}\label{Delta_3_10}
\begin{aligned}
\Delta_{3}=&\mathbb{P}\left\{\alpha'_\text{B}<|g|^2<\alpha'_1,|h|^2<\frac{\tau'(|g|^2)}{\rho_{m}}\right\}+\mathbb{P}\left\{|g|^2>\alpha'_1,|h|^2<\alpha'_\text{F}\right\}.
\end{aligned}
\end{equation}
 Then the sum of $\Delta_{3}$, $\Delta_{4}$ and $\Delta_{5}$ can be expressed as
\begin{equation}\label{Delta_345}
\begin{aligned}
\Delta_{3}+\Delta_{4}+\Delta_{5}=&\mathbb{P}\left\{\alpha'_{\text{B}}<|g|^2<\alpha'_{1},|h|^2<\frac{\rho_{m}|g|^2+1}{\rho_m}\gamma_{\text{F}}\right\}\\
&+\mathbb{P}\left\{|g|^2>\alpha'_{1},|h|^2<\alpha'_\text{F}\right\}.
\end{aligned}
\end{equation}
Finally, by combining $\Delta_6$ and (\ref{Delta_345}), $\mathcal{P}_{\text{CS}}^{\text{PC}}$ can be expressed as
\begin{equation}\label{OP_PC_2}
\begin{aligned}
\mathcal{P}_{\text{CS}}^{\text{PC}}=
&\mathbb{P}\left\{|g|^2<\alpha'_{\text{B}},|h|^2<\alpha'_\text{F}\right\}+\mathbb{P}\left\{\alpha'_{\text{B}}<|g|^2<\alpha'_{1},|h|^2<(\rho_{m}|g|^2+1)\alpha'_\text{F}\right\}\\
&+\mathbb{P}\left\{|g|^2>\alpha'_{1},|h|^2<\alpha'_\text{F}\right\}.
\end{aligned}
\end{equation}

Following the same steps for the derivation of Theorem \ref{theorem2}, we can obtain (\ref{OP_theorem2}), and the proof is complete.

	
	
\end{document}